\documentclass[acmsmall]{acmart}
\makeatletter                   %1
\def\mdseries@tt{m}             %1
\makeatother                    %1
\usepackage[plain]{fancyref}
\usepackage[draft=true]{minted} %2
\usepackage{color}
\usepackage{hyperref}           %5
\hypersetup{
    colorlinks=true,
    linkcolor=blue,
    filecolor=red,      
    urlcolor=magenta,
    breaklinks=true,            %3
}
\usepackage{breakurl}           %3

\acmJournal{PACMPL}
\acmVolume{1}
\acmNumber{CONF} % CONF = POPL or ICFP or OOPSLA
\acmArticle{1}
\acmYear{2018}
\acmMonth{1}
\acmDOI{} % \acmDOI{10.1145/nnnnnnn.nnnnnnn}
\startPage{1}
\renewcommand\footnotetextcopyrightpermission[1]{}

\setcopyright{none}
\usepackage{booktabs}   %% For formal tables:
\usepackage{subcaption} %% For complex figures with subfigures/subcaptions
\usepackage{multirow}
\usepackage{booktabs}
\usepackage{dsfont}
\usepackage{footmisc}
\usepackage{lipsum, babel}
\usepackage{soul}
\usepackage{xspace}
\usepackage{algorithm}
\usepackage[noend]{algpseudocode}
\usepackage{esvect}

\usepackage{pgf}
\usepackage{tikz}
\usetikzlibrary{arrows,automata,positioning, shapes, shapes.geometric, fit, calc}
\usepackage[latin1, utf8]{inputenc}

\usepackage{proof}
\usepackage{enumitem}

\usepackage{chngcntr}
\usepackage{apptools}

\newcommand{\sys} [0] {{\it Rose}}

\newcommand{\alt}{~~|~~}

\newcommand{\binopdef}     \oplus % {\mathit{binop}}
\newcommand{\unopdef}      \ominus % {\mathit{unop}}

\makeatletter
\DeclareRobustCommand*\cal{\@fontswitch\relax\mathcal}
\makeatother

\def\cA{{\cal A}}

\def\cC{{\cal C}}
\def\cD{{\cal D}}

\def\cF{{\cal F}}

\def\cL{{\cal L}}

\def\cO{{\cal O}}
\def\cP{{\cal P}}

\def\cU{{\cal U}}

\def\cvp{{\varphi}}
\def\cVvp{{\vec{\cvp}}}

\newcommand{\exec} [1] {\llbracket #1 \rrbracket}

\newcommand{\tup} [1] {\langle #1 \rangle}

    \newcommand{\infral} [3] {\infer[\textsc{#3}]{\begin{array}{c} #2 \end{array} }{ \begin{array}{c} #1  \end{array} } }

        \newcommand{\eat} [1] {}

\makeatletter
\renewcommand\section{\@startsection{section}{1}{\z@}%
                       {-8\p@ \@plus -4\p@ \@minus -4\p@}%
                       {6\p@ \@plus 4\p@ \@minus 4\p@}%
                       {\normalfont\large\bfseries\boldmath
                        \rightskip=\z@ \@plus 8em\pretolerance=10000 }}
\renewcommand\subsection{\@startsection{subsection}{2}{\z@}%
                       {-8\p@ \@plus -4\p@ \@minus -4\p@}%
                       {6\p@ \@plus 4\p@ \@minus 4\p@}%
                       {\normalfont\normalsize\bfseries\boldmath
                        \rightskip=\z@ \@plus 8em\pretolerance=10000 }}
\renewcommand\subsubsection{\@startsection{subsubsection}{3}{\z@}%
                       {-4\p@ \@plus -4\p@ \@minus -4\p@}%
                       {-1.5em \@plus -0.22em \@minus -0.1em}%
                       {\normalfont\normalsize\bfseries\boldmath}}
\makeatother

\begin{document}

%% Title information
\title[Noisy Program Synthesis using Abstractions]{Inductive Program 
Synthesis over Noisy Datasets using Abstraction Refinement Based Optimization
}         %% [Short Title] is optional;
                                        %% when present, will be used in
                                        %% header instead of Full Title.
%\titlenote{with title note}             %% \titlenote is optional;
                                        %% can be repeated if necessary;
                                        %% contents suppressed with 'anonymous'
%\subtitle{Subtitle}                     %% \subtitle is optional
%\subtitlenote{with subtitle note}       %% \subtitlenote is optional;
                                        %% can be repeated if necessary;
                                        %% contents suppressed with 'anonymous'

%% Author information
%% Contents and number of authors suppressed with 'anonymous'.
%% Each author should be introduced by \author, followed by
%% \authornote (optional), \orcid (optional), \affiliation, and
%% \email.
%% An author may have multiple affiliations and/or emails; repeat the
%% appropriate command.
%% Many elements are not rendered, but should be provided for metadata
%% extraction tools.

%% Author with single affiliation.
\author{Shivam Handa}
%\authornote{with author1 note}          %% \authornote is optional;
                                        %% can be repeated if necessary
\orcid{nnnn-nnnn-nnnn-nnnn}             %% \orcid is optional
\affiliation{
  \department{EECS}              %% \department is recommended
  \institution{Massachusetts Institute of Technology}            %% \institution is required
  \country{USA}                    %% \country is recommended
}
\email{shivam@mit.edu}          %% \email is recommended

%% Author with two affiliations and emails.
\author{Martin Rinard}
%\authornote{with author2 note}          %% \authornote is optional;
                                        %% can be repeated if necessary
\orcid{nnnn-nnnn-nnnn-nnnn}             %% \orcid is optional
\affiliation{
  \department{EECS}             %% \department is recommended
  \institution{Massachusetts Institute of Technology}           %% \institution is required
  \country{USA}                   %% \country is recommended
}
\email{rinard@csail.mit.edu}         %% \email is recommended

%% Abstract
%% Note: \begin{abstract}...\end{abstract} environment must come
%% before \maketitle command
\begin{abstract}
    We present a new synthesis algorithm to solve program synthesis over noisy
datasets, i.e., data that may contain incorrect/corrupted input-output examples.
Our algorithm uses an abstraction refinement based optimization process to
synthesize programs which optimize the tradeoff between the loss over the noisy
dataset and the complexity of the synthesized program. 
The algorithm uses abstractions to divide the search space of programs into
subspaces by computing an abstract value that represents outputs for all programs
in a subspace. The abstract value allows our algorithm to compute, for each
subspace, a sound approximate lower bound of the loss over all programs in the subspace. 
It iteratively refines these abstractions to further subdivide the space into smaller 
subspaces, prune subspaces that do not contain an optimal program, and eventually
synthesize an optimal program. 

We implemented this algorithm in a tool called \sys.
We compare \sys{}~ to a current state-of-the-art noisy program synthesis
system~\cite{handa2020inductive} using the SyGuS 2018 benchmark
suite~\cite{alur2013syntax}. Our evaluation demonstrates that \sys~significantly
outperforms this previous system: 
on two noisy benchmark program synthesis problem sets drawn from the 
SyGus 2018 benchmark suite, \sys{}~ delivers speedups of up to 1587 and
81.7, with median speedups of 20.5 and 81.7. \sys{}~ also
terminates on 20 (out of 54) and 4 (out of 11) more benchmark problems
than the previous system. Both \sys{}~ and the previous system synthesize
programs that are optimal over the provided noisy data sets.  
For the majority of the problems in the benchmark
sets ($272$ out of $286$), the synthesized programs also 
produce correct outputs for all inputs in the original (unseen) noise-free
data set. These results highlight the benefits that \sys{}~ can 
deliver for effective noisy program synthesis.

\end{abstract}

%% 2012 ACM Computing Classification System (CSS) concepts
%% Generate at 'http://dl.acm.org/ccs/ccs.cfm'.

%% End of generated code

%% Keywords
%% comma separated list
  %% \keywords are mandatory in final camera-ready submission

%% \maketitle
%% Note: \maketitle command must come after title commands, author
%% commands, abstract environment, Computing Classification System
%% environment and commands, and keywords command.
\maketitle

\section{Introduction}
Program synthesis has been successfully used to synthesize programs
from examples, for domains such as string transformations~\cite{gulwani2011automating, 
singh2016transforming}, 
data wrangling~\cite{feng2017component}, data completion \cite{wang2017synthesis}, and
data structure manipulation~\cite{feser2015synthesizing, 
osera2015type, yaghmazadeh2016synthesizing}. 
In recent years, there has been interest in synthesizing programs
from input-output examples in presence of noise/corruptions~
\cite{handa2020inductive, handa2021program, raychev2016learning, peleg2020perfect}. 
This line of work aims to tackle real world datasets which contain noise and
corruptions. These techniques can synthesize the
correct programs, even in presence of substantial
noise~\cite{handa2020inductive}. 

Noisy program synthesis has been formulated as an optimization problem over the
program space and the noisy dataset~\cite{handa2020inductive, handa2021program}. 
The optimization problem is parameterized
with three functions: a {\it loss function}, which measures by how much 
a program's output differs from the output in the given noisy data set, 
a {\it complexity measure}, which measures
the complexity of a candidate program, and an {\it objective function}, which combines
these both of these scores to rank programs~\cite{handa2020inductive, handa2021program}. 
Given a search space and a set of noisy input/output pairs,
the task of the noisy program synthesis is to synthesize a program which
minimizes the objective function over the noisy input/output pairs.

Similar to noise-free program synthesis, noisy program synthesis is effectively a search problem. 
Working with noisy datasets further complicates the search: 
synthesizing a program which simply maximizes the number of input-output
examples it satisfies (is correct on) may not optimize the objective function
over the entire dataset. Moreover, recent research has demonstrated that
noisy program synthesis can often synthesize the correct program
even when all input/output examples are corrupted~\cite{handa2020inductive}.

Current solutions to noisy program synthesis, either fall into enumeration based
search techniques \cite{peleg2020perfect} or version space (finite tree
automaton) based techniques~\cite{handa2020inductive}.
Both of these techniques reduce the search space by partitioning the space of programs
based on their execution behaviour on the given inputs.
Both exploit the property that programs which produce the same output values on the
given input values will have the same loss on a given dataset. 
If we restrict our search space to a single partition (instead of all programs
within the given program space) then the simplest
program (based on the complexity measure) will be the optimal solution to the
noisy program synthesis problem. This is due to the fact that all programs
within this space have the same loss value. Therefore the program which minimizes
the complexity measure, minimizes the objective function.
If the simplest program within this partition is not the optimal program (i.e,
there exists another program in the search space which further reduces the value
of the objective function) then we can safely conclude no other program
within this partition is the optimal program. 
With this knowledge, iterating over all partitions and  
comparing their simplest program allows these techniques to synthesize the
optimal program.
The performance of these techniques is determined by the number of partitions
that are created, given a dataset.

This partitioning approach has also been used in traditional noise-free
program synthesis settings~\cite{gulwani2011automating, wang2017synthesis} and
horn-clause verification~\cite{kafle2015tree}. A potential solution to
reducing the number of partitions was proposed by Kafle at al. in the context of
horn clause verification using tree automata~\cite{kafle2015tree}.
Wang et al introduced this technique to noise-free
program synthesis~\cite{wang2017program}. The technique uses {\it abstract} output values 
to partition the partition the program space, instead of {\it concrete} output values.
An {\it abstract} value is a compact representation of
a set of {\it concrete} values. Both of these techniques associate partitions of 
their search space to abstract values.
For example, \cite{wang2017program} represents a partition with
an array of {\it abstract} output values and contains all programs
which, given input values, maps these inputs to an array of 
{\it concrete} outputs, where each 
output value is an element of the corresponding {\it abstract} output value.
Abstract values allow these techniques to reduce the number of partitions
and hence decrease the running time of the synthesis algorithm. 

Our work applies the {\it abstract} value based partitioning approach
to create an abstraction refinement based algorithm for solving
the noisy program synthesis problem. Our technique uses abstractions
to partition the program space, with the abstract value of each
partition enabling us to soundly and approximately 
estimate a lower bound on the loss value of a partition, i.e., 
the minimum loss of any program within that partition. 
This lower bound on the loss value, attached with the simplest program within a given
partition, allows us to effectively synthesize candidate optimal programs. 
Given a candidate optimal program, our technique can effectively prune out
partitions, which even with the minimum possible loss value, will fail to
contain the optimal program.
The remaining partitions, based on their lower bound loss value, may contain
programs which {\it better fit} the noisy dataset compared to our candidate
optimal program (i.e., the remaining partitions may contain programs that
the object function ranks above the candidate optimal program). 
The synthesis algorithm then refines these abstractions in order to further 
refine the abstraction-based partitioning and improve its estimate of the
minimum possible loss value.  
The algorithm {\bf guarantees} that it will eventually synthesize the optimal program.

We have implemented our algorithm in the \sys{}~ synthesis tool. 
\sys~can be instantiated to work in different domains by providing suitable domain
specific languages, abstract semantics, and concrete semantics 
of functions within the language. \sys~ is parameterized over 
a large class of objective functions, loss functions,
and complexity measures. \cite{handa2020inductive} highlights that this
flexibility is required to synthesize correct programs for datasets
which contain a large amount of noise. 

We \sys{} to a current state-of-the-art noisy program synthesis
system~\cite{handa2020inductive} using the SyGuS 2018 benchmark
suite~\cite{alur2013syntax}. Our evaluation demonstrates that \sys~significantly
outperforms this previous system: 
on two noisy benchmark program synthesis problem sets drawn from the 
SyGus 2018 benchmark suite, \sys{} delivers speedups of up to 1587 and
81.7, with median speedups of 20.5 and 81.7. \sys{}~ also
terminates on 20 (out of 54) and 4 (out of 11) more benchmark problems
than the previous system. Both systems synthesize programs that are
optimal over the provided noisy data sets. 
For the majority of the problems in the benchmark
sets ($272$ out of $286$), both systems also synthesize programs
that produce correct outputs for all inputs in the original (unseen) noise-free
data set. These results highlight the significant benefits that \sys{}~ can 
deliver for effective noisy program synthesis.

\noindent{\bf Contributions:} This paper makes the following key contributions:
\begin{itemize}
\item {\bf New Abstraction Technique:} It presents a new program synthesis technique for
synthesizing programs over noisy datasets. This
technique uses the abstract semantics of DSL constructs to 
        partition the program search space. For each partition, 
        the technique uses the abstract semantics to compute an
        abstract value representing the outputs of all programs in that
        partition. The abstract value also allows the technique to 
soundly estimate the minimum possible loss value over programs all programs in each partition. 

    \item {\bf New Refinement Technique:} It presents a new refinement technique that works
      with the sound approximation of the minimum loss values to refine the current
      partition, then discard partitions that cannot possibly contain the optimal 
      program. Iteratively applying this refinement technique delivers 
      the program with the optimal loss function over the given input/output examples. 

    \item {\bf \sys{}~Evaluation}: 
       It presents the \sys{} synthesis system, which implements the new abstraction
       and refinement techniques presented in this paper. Our experimental evaluation
       shows that \sys{}~ delivers substantial
       performance improvements over a current state of the art noisy program synthesis
       system~\cite{handa2020inductive} (Section~\ref{sec:results}). In comparison 
       with this previous system, these performance improvements result in substantially 
       smaller overall program synthesis times and many fewer synthesis timeouts. 
\end{itemize}

\section{Preliminaries}\label{sec:preliminaries}
We first review the noisy program synthesis framework (introduced
by~\cite{handa2020inductive, handa2021program}), the concepts associated with this framework, and
the conditions that qualify a program to be the correct solution to a synthesis
problem.
We also discuss the 
tree automata based noisy program synthesis technique proposed
by~\cite{handa2020inductive}.

\subsection{Finite Tree Automata}
{\it Finite Tree Automata} are a type of state machine which accept 
trees rather than strings. They generalize standard finite automata to 
describe a regular language over trees.

\begin{definition}[\bf FTA]
    A (bottom-up) Finite Tree Automaton (FTA) over alphabet $F$ is a tuple
    $\cA = (Q, F, Q_f, \Delta)$ where $Q$ is a set of states, $Q_f
    \subseteq Q$ is the set of accepting states and $\Delta$ is a set of
    transitions of the form $f(q_1, \ldots, q_k) \rightarrow q$ where 
    $q, q_1, \ldots q_k$ are states, $f \in F$.
\end{definition}

Every symbol $f$ in the set $F$ has an associated arity. The set 
$F_k \subseteq F$ is the set of all $k$-arity symbols in $F$.
$0$-arity terms $t$ in $F$ are viewed as single node trees (leaves of trees).
$t$ is accepted by an FTA if we can rewrite $t$ to some state $q \in Q_f$ using
rules in $\Delta$. We use the notation $\cL(\cA)$ to denote the language of an
FTA $\cA$, i.e., the set of all trees accepted by $\cA$.

\begin{example}
    Consider the tree automaton $\cA$ defined by states $Q = \{q_{T},
    q_{F}\}$, $F_0 = \{\mathsf{True}, \mathsf{False}\}$, 
    $F_1 = \mathsf{not}$, $F_2 = \{\mathsf{and}\}$, 
    final states $Q_f = \{q_{T}\}$ and the following transition rules $\Delta$:
    \[
        \begin{array}{cccc}
            \mathsf{True} \rightarrow q_T 
            &
            \mathsf{False} \rightarrow q_F
            &
            \mathsf{not}(q_T) \rightarrow q_F
            &
            \mathsf{not}(q_F) \rightarrow q_T
            \\
            \mathsf{and}(q_T, q_T) \rightarrow q_T
            &
            \mathsf{and}(q_F, q_T) \rightarrow q_F
            &
            \mathsf{and}(q_T, q_F) \rightarrow q_F
            &
            \mathsf{and}(q_F, q_F) \rightarrow q_F
            \\
            \mathsf{or}(q_T, q_T) \rightarrow q_T
            &
            \mathsf{or}(q_F, q_T) \rightarrow q_T
            &
            \mathsf{or}(q_T, q_F) \rightarrow q_T
            &
            \mathsf{or}(q_F, q_F) \rightarrow q_F
            \\
        \end{array}
    \]
\end{example}

The above tree automaton accepts all propositional logic formulas
over $\mathsf{True}$ and $\mathsf{False}$
which evaluate to $\mathsf{True}$. 
\eat{
Figure~\ref{fig:formula_tree} presents the tree for the formula
$\mathsf{and}(\mathsf{True}, \mathsf{not}(\mathsf{False}))$. 

\begin{figure}
    \begin{tikzpicture}[shorten >=1pt,node distance=1.25cm,on grid]
        \node[state]   (q1)                {$\mathsf{and}$};
        \node[state]           (q2) [above right=of q1] {$\mathsf{True}$};
        \node[state] (q3) [below right=of q1] {$\mathsf{not}$};
        \node[state] (q4) [right=of q3] {$\mathsf{False}$};
        \path[->] (q2) edge                node [] {} (q1)
                    (q3) edge node [] {} (q1)
                    (q4) edge node [] {} (q3);
\end{tikzpicture}
    \caption{Tree for formula $\mathsf{and}(\mathsf{True}, \mathsf{not}(\mathsf{False}))$}
    \label{fig:formula_tree}
\end{figure}
}

\subsection{Domain Specific Languages (DSLs)}
\label{sec:dsl} 
The noisy program synthesis framework uses a Domain Specific Language (DSL) to 
specify to the set of programs under consideration. Without loss of generality, we
assume all programs $p$ are of the form $\lambda~x.e$ (i.e., they take a single
input $x$), where $e$ is a parse trees in a Context Free Grammar $G$. 
The internal nodes of this parse tree represent function and the leaves
represent constants or variable $x$.
$\llbracket p \rrbracket x_j$ denotes the
output of $p$ on input value $x_j$ ($\llbracket . \rrbracket$ is
defined in Figure~\ref{fig:exec_sem}).

\begin{figure}
    \[
        \begin{array}{c}
            \begin{array}{cc}
            \infral{
                t \in T_C
            }
            {\llbracket t \rrbracket x_j \Rightarrow v_t}
            {(Constant)}
            &
            \infral{}
            {\llbracket x \rrbracket x_j \Rightarrow x_j}
            {(Variable)}
            \end{array}
            \\
            \\
            \infral{
            \llbracket e_1 \rrbracket x_j \Rightarrow v_1 ~~~~~ 
            \llbracket e_2 \rrbracket x_j \Rightarrow v_2 ~~~~ \ldots ~~~~
            \llbracket e_k \rrbracket x_j \Rightarrow v_k
            }
            {\llbracket f(e_1, e_2, \ldots e_k) \rrbracket x_j  \Rightarrow
            f(v_1, v_2, \ldots v_k)}{(Function)}
    \end{array}
    \]
    \caption{Execution semantics for program $p$}
    \label{fig:exec_sem}
\end{figure}

All valid programs (which can be executed) are defined by a DSL 
grammar $G = (T, N, P, s_0)$ where:
\begin{itemize}[leftmargin=*]
    \item $T$ is a set of terminal symbols.
        These include 
        constants and input variable $x$.
        %symbols which may change value depending on the input
        %$\sigma$.
        We use the notation $T_C$ to denote the set of constants in $T$.
    \item $N$ is the set of non-terminals that represent subexpressions in our
        DSL.
    \item $P$ is the set of production rules of the form $s \rightarrow f(s_1,
        \ldots, s_n)$, where $f$ is a built-in function in the DSL and 
        $s, s_1, \ldots, s_n$ are non-terminals in the grammar.
    \item $s_0 \in N$ is the start non-terminal in the grammar.
\end{itemize}

We assume that we are given a black box implementation of 
each built-in function $f$
in the DSL.  In general, all techniques explored within 
this paper can be generalized to
any DSL which can be specified within the above framework.

\begin{example}\label{ex:dsl}
    The following DSL defines expressions over input x, 
    constants 2 and 3, and addition and multiplication:
    \[
        \begin{array}{rcl}
            n &:=& x \alt n + t \alt n \times t; \\
            t &:=& 2 \alt 3;
        \end{array}
    \]
\end{example}

We use the notation $p[x]$ to denote $\exec{p}x$. Given a vector of input values
$\vec{x} = \tup{x_1, \ldots x_n}$, we use the notation $p[\vec{x}]$ to denote
the vector $\tup{p[x_1], \ldots p[x_n]}$.

\subsection{Loss Functions}\label{subsec:lossFunction}
Given a noisy input-output example $(x, y)$ 
and a program $p$, a {\bf Loss Function}
$L(p[x], y)$ measures how incorrect the program is with respect to the given
example. The loss function only depends on the noisy output $y$ and the output
of the program $p$ on the input $x$.
Given a noisy dataset $\cD = (\vec{x}, \vec{y})$ (where $\vec{x} = \tup{x_1,
\ldots x_n}$ and $\vec{y} = \tup{y_1, \ldots y_n}$), we use the notation
$L(p[\vec{x}], \vec{y})$ to denote the sum of loss of program $p$ on individual
corrupted input-output example, i.e.,
\[
    L(p[\vec{x}], \vec{y}) = \sum\limits_{i=1}^n L(p[x_i], y_i)
\]

%\noindent
\begin{definition}
{\bf $0/1$ Loss Function:} The $0/1$ loss function 
    $L_{0/1}(p[x], y)$ returns $0$, if $p$ agrees with an input-output example
    $(x, y)$, else it returns $1$:
   \[
        L_{0/1}(p[x], y) = 
        1 
        ~~\mathrm{if}~~ (y \neq p[x]) ~~\mathrm{else}~~ 0
    \]
    Given a noisy dataset $\cD$, $L_{0/1}$
    counts the number of
    input-output examples where $p$ does not agree within the data set. 
\end{definition}

%\noindent
\begin{definition}
{\bf $0/\infty$ Loss Function:} The $0/\infty$ loss function 
    $L_{0/\infty}(p[x], y)$ is 0
    if $p$ matches the output $y$ on input $x$ and $\infty$ otherwise:
    %if $p$ matches all outputs in the data set $\cD$ and
%$\infty$ otherwise:
    \[
        L_{0/\infty}(p[x], y) = 
        0 ~~\mathrm{if}~~ (y = p[x]) ~~\mathrm{else}~~~ \infty
    \]
\end{definition}

%\noindent 
\begin{definition}
{\bf Damerau-Levenshtein (DL) Loss Function:} The DL loss function 
    $L_{DL}(p[x], y)$ uses the {\it Damerau-Levenshtein} 
    metric~\cite{damerau1964technique}
to measure the distance between the output from the synthesized program and the
corresponding output in the noisy data set: 
\[
    L_{DL}(p[x], y) = 
    L_{p[x], y}\big(\left\vert p[x]\right\vert, \left\vert y \right\vert\big)
\]
where, $L_{a, b}(i, j)$ is the {\it Damerau-Levenshtein} metric~\cite{damerau1964technique}.
\end{definition}
    This distance metric counts the number of 
single character deletions, insertions, substitutions, or transpositions
required to convert one input string into another.
We use this loss function to work with input-output examples containing
human-provided text, 
as more than $80\%$ of all human misspellings 
are reported to be captured by a single one of these four 
operations~\cite{damerau1964technique}.

\subsection{Complexity Measure}
\label{subsec:complexityMeasure}
Given a program $p$, a {\bf Complexity Measure} $C(p)$ 
scores a program $p$ based on their complexity.
This score is independent of the noisy dataset $\cD$.
%ranks 
%programs independent of the input-output examples in the 
%data set $\cD$.  
Within the noisy program synthesis framework, complexity measure 
is used to trade off performance on 
the noisy data set vs. complexity
of the synthesized program.  
Formally, a complexity measure $C(p)$ maps each program $p$ to a real number.

As standard in noisy program synthesis literature~\cite{handa2020inductive}, we
work with complexity measure of the form $\mathrm{Cost}(p)$, which
computes the complexity of given program $p$ represented as a parse 
tree recursively as follows:
\[
\begin{array}{rcl}
\mathrm{Cost}(t) &=& \mathrm{cost}(t)\\
\mathrm{Cost}(f(e_1, e_2, \ldots e_k))  &=& \mathrm{cost}(f) + 
    \sum\limits_{i = 1}^k \mathrm{Cost}(e_i) 
\end{array}
\]
where $t$ and $f$ are terminals and  built-in functions in our DSL respectively.
A popular example of a complexity measure, expressed in this form, is
$\mathsf{Size}(p)$, which assigns the cost of a terminal and function as $1$ 
($\mathrm{cost}(t) = \mathrm{cost}(f) = 1$).

Given an FTA $\cA$, we can synthesize the minimum complexity  parse tree (as
measured by complexity measures of the form $\mathsf{Cost}(p)$) accepted by
$\cA$ using a dynamic programming based algorithm proposed
by~\cite{gallo1993directed} (see Section $6$ in~\cite{wang2017program} for more details).

\subsection{Objective Functions}
An {\bf Objective Function}, within the noisy program synthesis framework,
combines the complexity score of a program and it's loss over a noisy dataset,
to be a combined score. The objective function is used to compare different
programs in the DSL and synthesize the program which {\it best-fits} the noisy
dataset. Formally,
an objective function $U$ maps loss and complexity tuples $\tup{l, c}$ to
a totally ordered set, such that, for all $l$ and $c$, $U(\tup{l, c})$ is
{\bf monotonically non-decreasing} with respect to $c$ and $l$. 
\begin{definition}
{\bf Tradeoff Objective Function:}
    Given a tradeoff parameter $\lambda > 0$, 
    the tradeoff objective function $U_{\lambda}(\tup{l, c}) = l + \lambda c$.
    This objective function allows us to trade-off between complexity and loss,
    directing the synthesis algorithm to search for a simpler program, in lieu
    of increasing the loss over the dataset.
\end{definition}

\begin{definition}
{\bf Lexicographic Objective Function:}
   A lexicographic objective function
    $U_L(\tup{l,c}) = \tup{l,c}$ maps $l$ and $c$ 
   into a lexicographically ordered space, i.e., $\tup{l_1, c_1} < \tup{l_1, c_2}$ if and only
   if either $l_1 < l_2$ or $l_1 = l_2$ and $c_1 < c_2$.
This objective function first minimizes the loss, then the complexity. 
\end{definition}

\noindent We will use the
notation $\tup{l_1, c_1} \leq_U \tup{l_2, c_2}$ to denote $U(\tup{l_1, c_1})
\leq U(\tup{l_2, c_2})$.

\subsection{Program Synthesis over Noisy Datasets}
Given a noisy dataset  $\cD = (\vec{x}, \vec{y})$, the noisy programs synthesis
aims to find a program $p$ which {\it best-fits} the dataset $\cD$, where the
{\it best-fit} is defined by the loss function $L$, complexity measure $C$, and
objective function $U$. Formally,
given a DSL $G$, we wish to find a program
$p^* \in G$ which minimizes the objective function $U$ (parameterized by $L$ and
$C$), i.e.,  
        \[
    %\\
  %\iff  
        %\\
        p^* \in \mathsf{argmin}_{p \in G} U(L(p[\vec{x}], \vec{y}), C(p))
%    \end{array}
\]

The above condition is equivalent to:
\[
 %   \begin{array}{c}
    \forall~p \in G.~\tup{ L(p^*[\vec{x}],\vec{y}), C(p)} \leq_U 
    \tup{ L(p[\vec{x}], \vec{y}), C(p)}  
\]

\noindent{\bf Concrete Finite Tree Automata:}
The noisy synthesis algorithm introduced by~\cite{handa2020inductive} builds upon the concept
of a Concrete Finite Tree Automaton (CFTA). Given a domain specific language $G$
and inputs $\vec{x}$, a Concrete Finite Tree Automaton is a tree automaton which
accepts all abstract syntax trees representing DSL programs. 
A CFTA is
constructed using the rules in Figure~\ref{fig:rulesfta}.
The states of the
CFTA is of the form $q^{\vec{v}}_{s}$, where $s$ is a symbol in $G$ and
$\vec{v}$ is vector of values. The existence of a state $q^{\vec{v}}_s$ implies
there exists a sub-expression in $G$, starting from symbol $s$, which maps
inputs $\vec{x}$ to output $\vec{v}$.
There exists a transition $f(q^{\vec{v}_1}_{s_1}, \ldots q^{\vec{v}_k}_{s_k})
\rightarrow  q^{\vec{v}}_s$
in the CFTA, only if,
for all $i = [1, |\vec{v}|]. f(v_{1i}, \ldots v_{ki}) = v_i$.

The Var Rule states that if we have an input symbol $x$, we construct a state
$q^{\vec{x}}_x \in Q$, where $\vec{x}$ is the input vector. 
The Const Rule states that for any constant (non variable terminal), we
construct a state 
$q^{\vec{v}}_t \in Q$, where $\vec{v}$ is a vector of size equal to $\vec{x}$
with each entry as $\exec{t}$.
The Final Rule states that, given a start symbol $s_0$, all states with symbol
$s_0$ are added to $Q_f$. The Prod Rule states that if we have a production $s
\rightarrow f(s_1, \ldots s_n) \in P$, and there exists states 
$q^{\vec{v}_1}_{s_1}, \ldots q^{\vec{v}_k}_{s_k} \in Q$, then we add a state
$q^{\vec{v}}_s \in Q$, where
for all $i = [1, |\vec{v}|]. f(v_{1i}, \ldots v_{ki}) = v_i$.
We also add a transition 
$f(q^{\vec{v}_1}_{s_1}, \ldots q^{\vec{v}_k}_{s_k})
\rightarrow  q^{\vec{v}}_s$ in the transition set $\Delta$.

Given a CFTA $(Q, Q_f, \Delta)$ (constructed using rules in
Figure~\ref{fig:rulesfta}) and a state
$q^{\vec{v}}_{s_0} \in Q_f$, a tree representing program $p$ is accepted by
automaton $(Q, \{q^{\vec{v}}_{s_0}\}, \Delta)$ if and only if $p[\vec{x}] =
     \vec{v}$ (\cite{handa2020inductive}).

\begin{figure}
\[
\begin{array}{c}
\begin{array}{ccc}
\infral{
    \vec{v}= \tup{
x_1, \ldots x_n
        }
}
{
    q^{\vec{v}}_x \in Q
}
{(Var)}
    &
\infral{
    q^{\vec{v}}_{s_0} \in Q
}
    {q^{\vec{v}}_{s_0} \in Q_f}
{(Final)} 
%\end{array}
%\\
%\\
    &
    \infral{
        t \in T_C, \vec{v} = \tup{
            \exec{t}, \ldots \exec{t}
        }, |\vec{v}| = n
    }
    {
        q^{\vec{v}}_t \in Q
    }
    {(Const)}
\end{array}
    \\
\\
    %\begin{array}{cc}
     %  &
    \infral{
    s \rightarrow f(s_1, \ldots, s_k) \in P, q^{\vec{v}_1}_{s_1}, \ldots,
    q^{\vec{v}_k}_{s_k} \in Q,  
    v_j = \exec{f(\vec{v}_{1j}, \ldots, \vec{v}_{kj})},
        \vec{v} = \tup{v_1, \ldots v_n}
}
    {q^{\vec{v}}_s \in Q, f(q^{\vec{v}_1}_{s_1}, \ldots, q^{\vec{v}_k}_{s_k}) \rightarrow
    q^{\vec{v}}_s \in \Delta}
{(Prod)}
    %\end{array}
\end{array}
\]
    \caption{Rules for constructing FTA $\cA = (Q, Q_f, \Delta)$ for inputs
    $\vec{x} = \tup{x_1, \ldots x_n}$.}
\label{fig:rulesfta}
\end{figure}

In general, the rules in Figure~\ref{fig:rulesfta} may result in a FTA which has
infinitely many states. \cite{handa2020inductive} handles this by only adding a new state
within the constructed FTA if the height of the smallest tree it will accept is 
less than some threshold height.

\noindent{\bf Synthesis Algorithm:} Figure~\ref{alg:oldalgorithm} 
presents a simplified
version of the noisy synthesis algorithm presented in \cite{handa2020inductive}.
The algorithm, given a noisy dataset $\cD = (\vec{x}, \vec{y})$, a DSL $G$, 
     an objective function $U$, a loss function $L$, and a
complexity measure $C$, synthesizes a program which minimizes the objective
function, i.e., the program $p^*$ returned by the algorithm satisfies the
following constraint:
\[
    p^* \in \mathsf{argmin}_{p \in G} U(L(p[\vec{x}], \vec{y}), C(p))
\]
The algorithm first constructs a Concrete Finite Tree Automaton (line 2) based
on the rules presented in Figure~\ref{fig:rulesfta}. Given this Concrete Finite
Tree Automaton $(Q, Q_f, \Delta)$, 
the algorithm finds the least complex program (i.e., program
which minimizes the complexity measure) for each accepting state $q \in Q_f$
(line 3-4). Given an accepting state of the form $q^{\vec{v}}_{s_0}$, a program
$p \in G$ is accepted by the automaton $(Q, \{q^{\vec{v}}_{s_0}\}, \Delta)$ 
if and only if $p[\vec{x}] = \vec{v}$.
Given an accepting $q^{\vec{v}}_{s_0}$, $P[q^{\vec{v}}_{s_0}]$ is the least
complex program which maps input vector $\vec{x}$ to outputs $\vec{v}$, i.e.,
\[
    P[q^{\vec{v}}_{s_0}] \in \mathsf{argmin}_{p \in G[\vec{x} \rightarrow
    \vec{v}]} C(p)
\]
where
$
   G[\vec{x} \rightarrow \vec{v}] = \{p \vert p \in G, p[\vec{x}] = \vec{v}\}
$.

The algorithm then finds an accepting state $q^{\vec{v}^*}_{s_0} \in Q_f$, 
such that,
for all accepting states $q^{\vec{v}}_{s_0} \in Q_f$,
\[
   \tup{L(\vec{v}^*, \vec{y}), C(P[q^{\vec{v}^*}_{s_0}])}
\leq_U
   \tup{L(\vec{v}, \vec{y}), C(P[q^{\vec{v}}_{s_0}])}
\]
Since, for all programs $p \in G$, there exists an accepting state
$q^{\vec{v}}_{s_0}$, such that, $p$ is accepted by 
$(Q, \{q^{\vec{v}}_{s_0}\}, \Delta)$ (\cite{handa2020inductive}),
the following statement is true:
\[
       \tup{L(\vec{v}^*, \vec{y}), C(P[q^{\vec{v}^*}_{s_0}])}
       \leq_U
   \tup{L(\vec{v}, \vec{y}), C(P[q^{\vec{v}}_{s_0}])}
\leq_U
   \tup{L(\vec{v}, \vec{y}), C(p)}
\]
Therefore, 
\[
    P[q^{\vec{v}^*}_{s_0}] \in \mathsf{argmin}_{p \in G} U(L(p[\vec{x}],
    \vec{y}), C(p))
\]

\begin{figure}[tb]
\begin{algorithmic}[1]
    \Procedure{Synthesize}{$\cD, G, U, L, C$}
    \Statex {\bf input:} 
    Noisy Dataset $\cD = (\vec{x}, \vec{y})$, DSL $G$.
\Statex{\bf input:} 
    Objective Function $U$, Loss Function $L$, and Complexity metric $C$.
    \Statex {\bf output:} A program $p^*$, such that, $\forall~p \in G.
\tup{L(p^*[\vec{x}], \vec{y}), C(p^*)} \leq_U  \tup{L(p[\vec{x}], \vec{y}),
C(p)}$.
\State $(Q, Q_f, \Delta) := \mathsf{ConstructFTA}(\vec{x}, G)$; 
\For{$q \in Q_f$} 
\State $P[q] := \mathsf{LeastComplex}_C((Q, \{q\}, \Delta))$; 
\Statex{\Comment{Least complex program for a given accepting state.}
}
\EndFor
\State $q^* := \mathsf{null}$; $\vec{v}^* := \mathsf{null};$
\For{$q^{\vec{v}}_{s_0} \in Q_f$}
\If{$q^* = \mathsf{null}$ or $\tup{L(\vec{v}^*, \vec{y}), C(P[q^*])} 
\leq_U \tup{L(\vec{v}, \vec{y}), C(P[q^{\vec{v}}_{s_0}])}$
}
\State $\vec{v}^* = \vec{v}$; $q^* := q^{\vec{v}}_{s_0}$;
\EndIf
\EndFor
\Statex{\Comment{$q^*$ is the accepting state which accepts a program which
minimizes the Objective Function.}}
\State {\bf return} $P[q^*]$;
\EndProcedure
\end{algorithmic}
\caption{Algorithm for noisy program synthesis using Finite Tree Automaton.}
\label{alg:oldalgorithm}
\end{figure}

\section{Noisy Program Synthesis using Abstraction Refinement Based Optimization}
We next introduce our synthesis algorithm which builds upon the concept of a
Finite Tree Automaton which uses abstract values instead of concrete 
values, as used in Figure~\ref{alg:oldalgorithm}.
The algorithm then uses an 
abstraction refinement based optimization technique to direct the algorithm 
towards the program which optimizes the objective function.

\subsection{Abstractions}\label{subsec:abstractions}
We construct an abstract version of the Finite Tree Automaton by associating 
abstract values with each symbol. 
We assume that the abstract values are
represented as 
conjunctions of predicates of the form $f(s)~\mathsf{op}~c$, where $s$ is a
symbol in the given DSL, $f$ is a function, $\mathsf{op}$ is an operator, and
$c$ is a constant.

\noindent{\bf Universe of predicates:} 
Given a DSL, our algorithm is parameterized by 
a universe $\cU$ of predicates, that our algorithm uses to construct
abstractions for our synthesis algorithm. 
The universe $\cU$ is specified using a 
family of function $\cF$, a set of operators $\cO$, and a set of constants
$\cC$, such that, all predicate in the universe $\cU$ can be written as
$f(s)~\mathsf{op}~c$, where $f \in \cF$, $\mathsf{op} \in \cO$, $c \in \cC$,
and $s$ is a symbol in the DSL (except predicates $\mathsf{true}$ and
$\mathsf{false}$).
We assume that $\cF$ contains the identity function, $\cO$ contains equality,
and $\cC$ includes the set of all values that can be computed by any
sub-expression within the DSL $G$. 

\noindent{\bf Notation:} Given predicates $\cP \subseteq \cU$ and an abstract
value $\cvp \in \cU$, we use $\alpha^\cP(\cvp)$ to denote the strongest
conjunctions of predicates in $\cP$, such that $\cvp \implies \alpha^\cP(\cvp)$.
Given a vector of abstract values $\cVvp = \tup{\cvp_1, \ldots \cvp_n}$, 
$\alpha^\cP(\cVvp)$ denotes the vector $\tup{\alpha^\cP(\cvp_1), \ldots
\alpha^\cP(\cvp_n)}$. 
As standard in abstract interpretation literature~\cite{cousot1977abstract},
we use the notation $\gamma(\cvp)$ to denote the set
of concrete values represented by the abstract value $\cvp$.

\noindent{\bf Abstract semantics:} In addition to concrete semantics for each
DSL construct, we are given
abstract semantics of each DSL construct in the form of symbolic
post-conditions over the universe of predicates $\cU$. Given a production $s
\rightarrow f(s_1, \ldots, s_n)$, we use the notation $\exec{f(\cvp_1, \ldots
\cvp_k)}^\#$ to represent the abstract semantics of function $f$, i.e.,
$\exec{f(\cvp_1, \ldots, \cvp_k)}^\# = \cvp$ if the function $f$ returns $\cvp$
(for symbol $s$),
given abstract values $\cvp_1, \ldots, \cvp_k$ for arguments $s_1, \ldots, s_k$. 
We assume that the abstract semantics are {\bf sound}, i.e.,
\[
    \exec{f(\cvp_1, \ldots, \cvp_k)}^\# = \cvp~\mathsf{and}~v_1 \in
    \gamma(\cvp_1), \ldots v_k \in \gamma(\cvp_k) \implies \exec{f(v_1, \ldots
    v_k)}
    \in \gamma(\cvp)
\]
However, we do not require the abstract semantics to be {\bf
precise}, i.e., formally:
\[
    \exec{f(\cvp_1, \ldots, \cvp_k)}^\# = \cvp, \text{there may exist a}~v \in \cvp, s.t., 
     \nexists~v_1 \in \cvp_1, \ldots v_k \in \cvp_k. \exec{f(v_1, \ldots v_k)} =
     v
\]
There may exist concrete value $v$ in the abstract output $\cvp$, such that, no
concrete input parameters $v_1, \ldots v_k$ in the abstract inputs $\cvp_1,
\ldots \cvp_k$ exist, for which $f(v_1, \dots v_k) = v$.

But we require the abstract semantics to be precise if all of the input
parameters are abstract values representing a single concrete value, i.e., 
they are abstract
values of the form $s = v$.
Formally:
\[
    \exec{f(s_1 = v_1, \ldots s_k = v_k)}^\# = (s = \exec{f(v_1, \ldots v_k)})
\]
Given a program $p$, predicates $\cP$, and input $x_j$, $\exec{p}^\cP x_j$
denotes the abstract value of program $p$, if the intermediate computed values
are only represented via predicates in $\cP$.
Figure~\ref{fig:execabstract} presents the precise rules for computing
$\exec{p}^\cP x_j$.
Given inputs $\vec{x} = \tup{x_1, \ldots x_n}$, we use the notation
$\exec{p}^\cP \vec{x}$ to denote the vector $\cVvp = \tup{\exec{p}^\cP x_1,
\ldots \exec{p}^\cP x_n}$.

\begin{figure}
    \[
        \begin{array}{c}
            \begin{array}{cc}
            \infral{
                t \in T_C
            }
                { \llbracket t \rrbracket^\cP x_j \Rightarrow \alpha^\cP(t =
                \exec{t}x_j)}
            {(Constant)}
            &
            \infral{}
                {\llbracket x \rrbracket^\cP x_j \Rightarrow 
                \alpha^\cP(x = x_j)}
            {(Variable)}
            \end{array}
            \\
            \\
            \infral{
            \llbracket e_1 \rrbracket x_j \Rightarrow \cvp_1 ~~~~~ 
            \llbracket e_2 \rrbracket x_j \Rightarrow \cvp_2 ~~~~ \ldots ~~~~
            \llbracket e_k \rrbracket x_j \Rightarrow \cvp_k
            }
            {\llbracket f(e_1, e_2, \ldots e_k) \rrbracket x_j  \Rightarrow
            \alpha^\cP(\exec{f(\cvp_1, \cvp_2, \ldots \cvp_k)}^\#)}{(Function)}
    \end{array}
    \]
    \caption{Abstract Execution semantics for program $p$.}
    \label{fig:execabstract}
\end{figure}

\noindent{\bf Abstract Loss Function:} 
Given a loss function, the abstract semantics of a loss function allows us to 
find the minimum possible loss value for a given abstract value, i.e., 
given a loss function $L$, noisy output $y$,
and an abstract value $\cvp$:
\[
    L(\cvp, y) = \min\limits_{z \in \gamma(\cvp)} L(z, y)
\]

\subsection{Abstract Finite Tree Automaton}
An Abstract Finite Tree Automata (AFTA) generalizes an Concrete Finite Tree
Automaton by replacing concrete values by abstract values while constructing
automaton states. This allows us to compress the size of an FTA, as multiple
states with concrete values can be represented
by a single
state with abstract value.

Given predicates $\cP$, DSL $G$, and
inputs $\vec{x}$, Figure~\ref{fig:rulesafta} presents the rules for constructing
an AFTA $(Q, Q_f, \Delta)$.
States in an AFTA are of the form $q^{\cVvp}_{s}$, where $s$ is a symbol and
$\cVvp$ is a vector of abstract values. 
If $q^{\cVvp}_s \in Q$, then there exists an expression $e$, starting from
symbol $s$, such that $\exec{e}^\cP\vec{x} = \cVvp$.
If there is a transition $f(q^{\cVvp_1}_{s_1}, \ldots q^{\cVvp_k}_{s_k})
\rightarrow q^\cVvp_s$ in the AFTA then
\[
    \forall j = [1, |\cVvp|]. \exec{f(\cvp_{1j},\ldots \cvp_{kj})}^\# \implies
    \cvp_j\]
The Var Rule states constructs a state
$q^{\cVvp}_x \in Q$ for variable symbol $x$, where 
$\cVvp = \alpha^\cP((x = x_1), \ldots (x = x_n))$.
The Const Rule constructs a state 
$q^{\cVvp}_t \in Q$ for each constant terminal $t$, 
where $\cVvp$ is a vector of size equal to $\vec{x}$
with each entry as $\alpha^\cP(t = \exec{t})$.
The Final Rule adds all states with symbol
$s_0$ are added to $Q_f$, where $s_0$ is the start symbol. The Prod Rule 
constructs a state
$q^{\cVvp}_s \in Q$, if there exists a production $s
\rightarrow f(s_1, \ldots s_n) \in P$ and there exists states 
$q^{\cVvp_1}_{s_1}, \ldots q^{\cVvp_k}_{s_k} \in Q$ (where
$\forall i = [1, |\vec{v}|]. f(\cvp_{1i}, \ldots \cvp_{ki}) \implies \cvp_i$).
We also add a transition 
$f(q^{\cVvp_1}_{s_1}, \ldots q^{\cVvp_k}_{s_k})
\rightarrow  q^{\cVvp}_s$ in $\Delta$.

\begin{figure}
\[
\begin{array}{c}
\begin{array}{cc}
\infral{
    \cVvp = \alpha^\cP\big(\tup{
x = x_1, \ldots x = x_n
        }\big)
}
{
    q^{\cVvp}_x \in Q
}
{(Var)}
    &
\infral{
    q^{\vec{\cvp}}_{s_0} \in Q
}
    {q^{\vec{\cvp}}_{s_0} \in Q_f}
{(Final)} 
\end{array}
\\
\\
    \infral{
        t \in T_C, \cVvp = \alpha^\cP\big(\tup{
            t = \exec{t}, \ldots t = \exec{t}
        }\big), |\cVvp| = n
    }
    {
        q^\cVvp_t \in Q
    }
    {(Const)}
    \\
\\
    %\begin{array}{cc}
     %  &
    \infral{
    s \rightarrow f(s_1, \ldots, s_k) \in P, q^{\cVvp_1}_{s_1}, \ldots,
    q^{\cVvp_k}_{s_k} \in Q,  
    \cvp_j = \alpha^\cP\big(\exec{f(\cVvp_{1j}, \ldots, \cVvp_{kj})}^\#\big),
        \cVvp = \tup{\cvp_1, \ldots \cvp_n}
}
    {q^{\cVvp}_s \in Q, f(q^{\cVvp_1}_{s_1}, \ldots, q^{\cVvp_k}_{s_k}) \rightarrow
    q^{\cVvp}_s \in \Delta}
{(Prod)}
    %\end{array}
\end{array}
\]
    \caption{Rules for constructing FTA $\cA = (Q, F, Q_f, \Delta)$ with
    Abstract Values, for inputs
    $\vec{x} = \tup{x_1, \ldots x_n}$.}
\label{fig:rulesafta}
\end{figure}

\begin{theorem}{\bf (Structure of the Tree Automaton)}
    Given a set of predicates $\cP$, input vector $\vec{x} = \tup{x_1, \ldots
    x_n}$, and DSL $G$,
    let $\cA = (Q, Q_f, \Delta)$ be the AFTA returned by the
    function $\mathsf{ConstructAFTA}(\vec{x}, G, \cP)$.
    Then for all symbols $s$ in $G$, 
    for all expressions $e$ starting from symbol $s$ 
    (and height less than bound $b$), 
    there exists a state $q^{\cVvp}_s \in Q$,
    such that, $e$ is accepted by the automaton $(Q, \{q^\cVvp_s\}, \Delta)$, 
    where $\cVvp = \tup{\exec{e}^\cP x_1, \ldots \exec{e}^\cP x_n }$. 
%    For any 
%    state $q^{\cVvp}_{s} \in Q$, a subexpression $e$ is accepted
%    by the automaton $(Q, \{q^{\cVvp}_{s}\}, \Delta)$ 
%    if and only if $\forall i \in [1, n].(s = \exec{e}x_i) \implies \cvp_i$.
\end{theorem}
We present the proof of this theorem in the
appendix~\ref{sec:appendixproofs}(Theorem~\ref{thm:appendixstructure}).

\begin{corollary}\label{cor:abstractvalue}
    Given a set of predicates $\cP$, input vector $\vec{x} = \tup{x_1, \ldots
    x_n}$, and DSL $G$,
    let $\cA = (Q, Q_f, \Delta)$ be the AFTA returned by the
    function $\mathsf{ConstructAFTA}(\vec{x}, G, \cP)$.
    All programs $p$ (of height less than bound $b$) 
    is accepted by $\cA$. 
    For any 
    accepting state $q^{\cVvp}_{s_0} \in Q_f$, a program $p$ is accepted
    by the automaton $(Q, \{q^{\cVvp}_{s_0}\}, \Delta)$ 
    if and only if $\forall i \in [1, n]$,
    $p[x_i] \in \gamma(\cvp_i)$.
\end{corollary}

\subsection{Synthesis Algorithm}
We present our synthesis algorithm in
Figure~\ref{alg:top-level}.
The
$\mathsf{Synthesize}$ procedure takes a noisy dataset $\cD$, a DSL $G$, 
a threshold $\epsilon \geq 0$, initial predicates $\cP$, a universe of
possible predicates $\cU$, objective function $U$, loss function $L$, and a
complexity measure $C$. We assume that $\mathsf{true}, \mathsf{false} \in \cP$.

The synthesis algorithm consists of a refinement loop (line 3-10). The loop
first constructs a Abstract Finite Tree Automaton (line 4) with the current set
of predicates $\cP$ using rules presented in Figure~\ref{fig:rulesafta}. The
algorithm then uses the $\mathsf{MinCost}$ function to generate a candidate
program $p^*$ (line 5). 
The algorithm maintains the program $p_r$, which is the {\it best} candidate program 
out of all the candidate programs $p^*$ generated.

If the {\it distance} between the current program $p^*$ and the best possible program
in the DSL $G$ is less a tolerance level $\epsilon$ ($\mathsf{Distance}$
function, line 8), the algorithm returns the best candidate program $p_r$. Note that
\[
    \tup{L(p_r[\vec{x}], \vec{y}), C(p_r)} \leq_U \tup{L(p^*[\vec{x}], \vec{y}),
    C(p^*)}
\]
Otherwise, the algorithm refine our AFTA to either improve our estimation of the best
possible program or synthesize a better candidate program. To refine our AFTA,
the algorithm first picks an input-output example $(x, y)$ from dataset $\cD$,
on which we can improve the candidate program $p^*$ (line 9). Given an
input-output example $(x, y)$, the procedure $\mathsf{OptimizeAndBackPropogate}$
constructs the constraints required to improve the AFTA and then returns the set
of predicates which will allow the algorithm to build a more refined AFTA.

We discuss each of these sub-procedures in detail next.
\begin{figure}[tb]
\begin{algorithmic}[1]
    \Procedure{Synthesize}{$\cD, G, \epsilon, \cP, \cU, U, L, C$}
    \Statex {\bf input:} 
    Noisy Dataset $\cD = (\vec{x}, \vec{y})$, DSL $G$, and tolerance $\epsilon$.
    \Statex {\bf input:}
    initial predicates $\cP$, and universe of predicates $\cU$.
    \Statex{\bf input:}
    Objective Function $U$, Loss Function $L$, and Complexity metric $C$.
    \Statex {\bf output:} A program $p^*$, such that, $p^*$ satisfies the 
    {\it $\epsilon$-correctness} condition.
    %\Statex $\forall p \in G. L(p^*[\vec{x}], \vec{y}) - L(p[\vec{x}], \vec{y}) >
    %\epsilon \implies 
%\tup{L(p^*[\vec{x}], \vec{y}), C(p^*)} \leq_U  \tup{L(p[\vec{x}], \vec{y}),
%C(p)}
%    $
    \State $p_r := \mathsf{null}$;
   \While{true}
    \State $\cA := \mathsf{ConstructAFTA}(\vec{x}, G, P)$; 
   \State $p^*:= \mathsf{MinCost}(\cA, \cD, U, L, C)$; 
    \If{$p_r = \mathsf{null}$ or $\tup{L(p^*[\vec{x}], \vec{y}), C(p^*)} \leq_U 
    \tup{L(p_r[\vec{x}], \vec{y}), C(p_r)}$}
    \State $p_r := p^*$;
    \EndIf 
    \If{$ \mathsf{Distance}(p^*, \cD, \cP, L) \leq \epsilon$} {\bf return} $p_r$;
    \EndIf
    \State $x, y := \mathsf{PickDimension}(p^*, \cD, \cP, L)$;
    \State $\cP := \cP 
    \bigcup \mathsf{OptimizeAndBackPropogate}(p^*, x, y, \cP, \cU)$;
   % \State $\cP := \cP \bigcup \mathsf{ExtractPredicates}(\cI)$; 
    \EndWhile
\EndProcedure
\end{algorithmic}
\caption{Algorithm for noisy program synthesis using Abstraction Refinement
    based Optimization.}
\label{alg:top-level}
\end{figure}

\subsection{\bf Minimum Cost Candidate}
We present the implementation of procedure $\mathsf{MinCost}$ in
Figure~\ref{alg:mincost}.
Given an AFTA $(Q, Q_f, \Delta)$, noisy dataset $\cD = (\vec{x}, \vec{y})$,
objective function $U$, loss function $L$, and complexity measure $C$,
$\mathsf{MinCost}$ returns a program $p_{q^*}$ which minimizes the {\it
abstract} objective function, where, given a program $p$, the {\it abstract}
objective function is defined as 
\[
    U(L(\exec{p}^\cP\vec{x}, \vec{y}), C(p))
\]

\begin{figure}[tb]
\begin{algorithmic}[1]
    \Procedure{MinCost}{$\cA, \cD, U, L, C$}
    \Statex {\bf input:} 
    AFTA $\cA = (Q, Q_f, \Delta)$,
    Noisy Dataset $\cD = (\vec{x}, \vec{y})$.
\Statex{\bf input:} 
    Objective Function $U$, Loss Function $L$, and Complexity metric $C$.
    \Statex {\bf output:} A program $p^*$, such that, $\forall~p \in G.
    \tup{L(\exec{p^*}^\cP\vec{x}, \vec{y}), C(p^*)} \leq_U
    \tup{L(\exec{p}^\cP\vec{x}, \vec{y}),
C(p)}$.
%\State $(Q, Q_f, \Delta) := \mathsf{ConstructFTA}(\vec{x}, G)$; 
\For{$q \in Q_f$} 
\State $P[q] := \mathsf{LeastComplex}_C((Q, \{q\}, \Delta))$; 
\Statex{\Comment{Least complex program for a given accepting state.}
}
\EndFor
\State $q^* := \mathsf{null}$; $\cVvp^* := \mathsf{null};$
\For{$q^{\cVvp}_{s_0} \in Q_f$}
\If{$q^* = \mathsf{null}$ or $\tup{L(\cVvp^*, \vec{y}), C(P[q^*])} 
\leq_U \tup{L(\cVvp, \vec{y}), C(P[q^{\cVvp}_{s_0}])}$
}
    \State $\cVvp^* = \cVvp$; $q^* := q^{\cVvp}_{s_0}$;
\EndIf
\EndFor
\Statex{\Comment{$q^*$ accepts a program which
minimizes the Abstract Objective Function.}}
\State {\bf return} $P[q^*]$;
\EndProcedure
\end{algorithmic}
\caption{Procedure for synthesizing the program which minimizes the Abstract
    Objective Function.}
\label{alg:mincost}
\end{figure}

\eat{
\[
\begin{array}{c}
    p_{q^*} = 
    \mathsf{MinCost}\big((Q, Q_f, \Delta), (\vec{x}, \vec{y}), U, L, C\big)
\\
    \mathsf{where}
~\forall~q \in Q_f. p_q := \mathsf{LeastComplex}_C(Q, \{q\}, \Delta)\\
q^* := \mathsf{argmin}_{q^{\cVvp}_{s_0} \in Q_f} U(L(\cVvp, \vec{y}), 
C(p_{q^\cVvp_{s_0}}))
\end{array}
\]
}

\noindent
Therefore, for all programs $p \in G$:
\[
    \tup{L(
    \exec{p_{q^*}}^\cP\vec{x}, \vec{y}
    ), C(p_{q^*})} \leq_U \tup{L(
    \exec{p}^\cP\vec{x}, \vec{y}
    ), C(p)} 
\]
Note that, since for all programs $p$, $L(\exec{p}^\cP\vec{x}, \vec{y}) \leq
L(p[\vec{x}], \vec{y})$, the following statement is true:
\[
\tup{L(
    \exec{p}^\cP\vec{x}, \vec{y}
    ), C(p)} \leq_U \tup{L(p[\vec{x}], \vec{y}), C(p)}
\]
Hence, for programs $p \in G$:
\[
    \tup{L(
    \exec{p_{q^*}}^\cP\vec{x}, \vec{y}
    ), C(p_{q^*})} \leq_U \tup{L(p[\vec{x}], \vec{y}), C(p)}
\]
The procedure first finds the least complex program (i.e., program
which minimizes the complexity measure) for each accepting state $q \in Q_f$
(line 2-3). Given an accepting state of the form $q^{\cVvp}_{s_0}$, a program
$p \in G$ is accepted by the automaton $(Q, \{q^{\cVvp}_{s_0}\}, \Delta)$ 
if and only if $\exec{p}^\cP\vec{x} = \cVvp$.
Given an accepting $q^{\cVvp}_{s_0}$, $P[q^{\cVvp}_{s_0}]$ is the least
complex program which maps input vector $\vec{x}$ to outputs $\cVvp$, i.e.,
\[
    P[q^{\cVvp}_{s_0}] \in \mathsf{argmin}_{p \in G[\vec{x} \rightarrow
    \cVvp]} C(p)
\]
where
$
   G[\vec{x} \rightarrow \cVvp] = \{p \vert p \in G, \exec{p}^\cP\vec{x} = \cVvp\}
$.

The algorithm then finds an accepting state $q^{\cVvp^*}_{s_0} \in Q_f$, 
such that,
for all accepting states $q^{\cVvp}_{s_0} \in Q_f$,
\[
   \tup{L(\cVvp^*, \vec{y}), C(P[q^{\cVvp^*}_{s_0}])}
\leq_U
   \tup{L(\cVvp, \vec{y}), C(P[q^{\cVvp}_{s_0}])}
\]

\begin{theorem}\label{thm:argmin}
    Given predicates $\cP$, DSL $G$, noisy dataset $\cD = (\vec{x}, \vec{y})$, 
    objective function $U$, loss
    function $L$, complexity measure $C$, and $\cA =
    \mathsf{ConstructAFTA}(\vec{x}, G, \cP)$, if $p^* =
    \mathsf{MinCost}(\cA, \cD, U, L, C)$ then 
    \[
        p^* \in \mathsf{argmin}_{p \in G} U(L(\exec{p}^\cP\vec{x}, \vec{y}),
        C(p))
    \]
    i.e., $p^*$ minimizes the {\it abstract} objective function.
\end{theorem}
We present the proof of this theorem in the appendix~\ref{sec:appendixproofs}
(Theorem~\ref{thm:appendixargmin}).
\eat{
\begin{proof}
    Let $\cA = (Q, Q_f, \Delta)$.
    From corollary~\ref{cor:abstractvalue}, for each program $p \in G$, there
    exists a state $q^\cVvp_{s_0} \in Q_f$, such that, $\exec{p}^\cP \vec{x} =
    \cVvp$.
Since the algorithm finds an accepting state $q^{\cVvp^*}_{s_0} \in Q_f$, 
such that,
for all accepting states $q^{\cVvp}_{s_0} \in Q_f$,
\[
   \tup{L(\cVvp^*, \vec{y}), C(P[q^{\cVvp^*}_{s_0}])}
\leq_U
   \tup{L(\cVvp, \vec{y}), C(P[q^{\cVvp}_{s_0}])}
\]
for all $p \in G$,
    \[
   \tup{L(\cVvp^*, \vec{y}), C(P[q^{\cVvp^*}_{s_0}])}
\leq_U
    \tup{L(\exec{p}^\cP\vec{x}, \vec{y}), C(p)}
\]
    Since $p^* = P[q^{\cVvp^*}_{s_0}]$,
     \[
        p^* \in \mathsf{argmin}_{p \in G} U(L(\exec{p}^\cP\vec{x}, \vec{y}),
        C(p))
    \]
\end{proof}
}

\subsection{\bf Termination Condition and Tolerance}
Given a candidate program $p^*$, predicates $\cP$, noisy dataset $\cD =
(\vec{x}, \vec{y})$, and a
loss function $L$, the $\mathsf{Distance}$ function returns the difference
between the concrete loss of program $p^*$ over noisy dataset $\cD$ and 
the abstract loss (given predicates
from $\cP$) over noisy dataset $\cD$. Formally:
\[
\mathsf{Distance}\big(p, 
(\vec{x}, \vec{y}), \cP, L\big)
:= 
%\sum\limits_{i=1}^n 
L(p[\vec{x}], \vec{y}) - L(\exec{p}^\cP \vec{x}, \vec{y})
\]

The algorithm terminates if the distance is less than equal to the tolerance
level $\epsilon$.
Note that if $\epsilon = 0$, then the algorithm only terminates when
$L(p^*[\vec{x}], \vec{y}) = L(\exec{p^*}^\cP\vec{x}, \vec{y})$, and since for
all program $p \in G$:
\[
    \tup{L(
    \exec{p_{q^*}}^\cP\vec{x}, \vec{y}
    ), C(p_{q^*})} \leq_U \tup{L(
    p[\vec{x}], \vec{y}
    ), C(p)} \]
is true, the following statement is also true:
\[
    p^* \in \mathsf{argmin}_{p \in G}U(L(p[\vec{x}], \vec{y}), C(p))
\]

In general, the previous work in noisy program synthesis have maintained a very
strict version of correctness, i.e., they generally synthesize a
program $p^*$ which minimizes the objective function. This leaves out any
speedups which can be achieved to by relaxing the requirement to synthesizing a
program which is {\it close} to the optimal program but may not be one of the
optimal programs.

To capture this relaxation, we introduce the concept of {\bf
$\epsilon$-correctness}. 
Given DSL $G$, a noisy dataset $\cD = (\vec{x}, \vec{y})$, an objective function
$U$, a loss function $L$, a complexity measure $C$, and set of programs $G_p
\subseteq G$, let $\mathds{B}_{\epsilon}(G_p)$ 
be a set of programs in $G$, 
such that, $p^* \in 
\mathds{B}_{\epsilon}(G_p)$ if and only if there exists
a program $p \in G_p$, such that, 
\[
    \tup{L(p^*[\vec{x}], \vec{y}), C(p^*)} \leq_U 
    \tup{L(p[\vec{x}], \vec{y}) + \epsilon, C(p)}
\]
i.e., $p^*$ will be a {\it better fit} the dataset $\cD$ compared to program
$p$, if the loss of $p$ over dataset $\cD$ was increased by $\epsilon$.  

\begin{definition}\label{def:epsiloncorrectness}{\bf ($\epsilon$-correctness)}
    Given a noisy dataset $\cD = (\vec{x}, \vec{y})$,
    a DSL $G$, an objective function $U$, a loss
    Function $L$, and a complexity measure $C$, a program $p_r \in G$ is
    {\bf $\epsilon$-correct} if and only if:
    \[
        p_r \in \mathds{B}_\epsilon(\mathsf{argmin}_{p \in G}U(L(p[\vec{x}],
        \vec{y}),C(p))
    \]
\end{definition}
\noindent
A program $p_r \in G$ is 
    {\bf $\epsilon$-correct} if and only if 
there exists a $p^* \in G$, such that,
    \[
        \tup{L(p_r[\vec{x}], \vec{y}), C(p_r)} \leq_U 
        \tup{L(p^*[\vec{x}], \vec{y}), C(p^*)}
    \]
    \[
      \forall p \in G.~\tup{L(p^*[\vec{x}] - \epsilon, \vec{y}), C(p^*)} 
      \leq_U \tup{L(p[\vec{x}], \vec{y}), C(p)}
    \]
Note that, for $\epsilon = 0$, the above condition reduces to
\[
    \forall p \in G.~\tup{L(p_r[\vec{x}], \vec{y}), C(p_r)} \leq_U
    \tup{L(p[\vec{x}], \vec{y}),C(p)}
\]
Therefore, for $\epsilon = 0$, $p_r \in \mathsf{argmin}_{p \in G}U(L(p[\vec{x}],
\vec{y}), C(p))$.

\begin{theorem}\label{thm:soundness}{\bf (Soundness)}
    Given a dataset $\cD$, a DSL $G$, tolerance $\epsilon \geq 0$, universe of
    predicates $\cU$, initial predicates $\cP$, objective function $U$, loss
    function $L$, and the complexity measure $C$, 
    if Algorithm~\ref{alg:top-level} returns the program $p^*$, then $p^*$
    satisfies the {\it $\epsilon$-correctness}
    condition (Definition~\ref{def:epsiloncorrectness}).
\end{theorem}
\begin{proof}
    Let us assume that the algorithm terminates on the $i^{th}$ iteration.
    Let $\cA_i = (Q, Q_f, \Delta)$ 
    be the AFTA when the algorithm terminates. Let $p_i$ be the
    program returned by $\mathsf{MinCost}$ on the $i^{th}$ iteration.

    \noindent
    From Theorem~\ref{thm:argmin}, for all programs $p \in G$,
\[
    \tup{L(
    \exec{p_i}^\cP\vec{x}, \vec{y}
    ), C(p_i)} \leq_U \tup{L(p[\vec{x}], \vec{y}), C(p)}
\]
    When the algorithm terminates,  the following condition is true:
    \[
        \mathsf{Distance}(p_i, \cD, \cP, L) \leq \epsilon
    \]
    which implies: 
    \[
        L(p_i[\vec{x}], \vec{y}) -  
        L(\exec{p_i}^{\cP_i} \vec{x}, \vec{y}) \leq \epsilon
    \]
    Therefore, for all programs $p \in G$, 
    \[
        \tup{L(p_i[\vec{x}], \vec{y}) - \epsilon, C(p_i)} \leq_U
        \tup{L(p[\vec{x}], \vec{y}), C(p)}
    \]
    and the following is true for the synthesized program $p_r$,
    \[
        \tup{L(p_r[\vec{x}], \vec{y}), C(p_r)} \leq_U
        \tup{L(p_i[\vec{x}], \vec{y}), C(p_i)}       
    \]

    \noindent
    Hence, if the algorithm~\ref{alg:top-level} returns a program $p_r$, then
    $p_r$ satisfies the {\it $\epsilon$-correctness} condition.
\end{proof}

\subsection{Abstraction Refinement based Optimization}
Given a dataset $\cD$ and predicates $\cP$, the program $p^*$ (returned by
$\mathsf{MinCost}$) minimizes the {\it abstract} objective function, i.e., for
all programs $p \in G$:
\[
    \tup{L(
    \exec{p^*}^\cP\vec{x}, \vec{y}
    ), C(p^*)} \leq_U 
 \tup{L(
    \exec{p}^\cP\vec{x}, \vec{y}
    ), C(p)} \leq_U 
    \tup{L(p[\vec{x}], \vec{y}), C(p)}
\]
Since, the algorithm did not terminate, $\mathsf{Distance}(p^*, \cP, \cD,
L) > \epsilon$.

Let us consider the case when $\epsilon = 0$. Since $\mathsf{Distance}(p^*, \cP,
\cD, L) > 0$, the concrete loss of program $p^*$ over dataset $\cD$ is greater
than the abstract loss of $p^*$ over $\cD$. Formally,
\[
    L(p^*[\vec{x}], \vec{y}) > L(\exec{p^*}^\cP \vec{x}, \vec{y})
\]
This means that
\[
   \tup{L(
    \exec{p^*}^\cP\vec{x}, \vec{y}
    ), C(p^*)} 
<_U 
    \tup{L(p^*[\vec{x}], \vec{y}), C(p^*)}
\]
At this point, even though, for all programs $p \in G$, the following is true:
\[
    \tup{L(
    \exec{p^*}^\cP\vec{x}, \vec{y}
    ), C(p^*)} \leq_U 
    \tup{L(p[\vec{x}], \vec{y}), C(p)}
\]
We cannot prove that $p^*$ is the optimal function, i.e.,
\[
   \tup{L(
    p^*[\vec{x}], \vec{y}
    ), C(p^*)} \leq_U 
    \tup{L(p[\vec{x}], \vec{y}), C(p)}
\]
And therefore, just using predicates $\cP$, we cannot prove that $p_r$ is the
optimal function.
\eat{
\[
    p_r \in \mathsf{argmin}_{p \in G} U(L(p[\vec{x}], \vec{y}), C(p))
\]
}

Similarly, if $\epsilon > 0$, for all programs $p \in G$, the following is true:
\[
    \tup{L(
    \exec{p^*}^\cP\vec{x}, \vec{y}
    ), C(p^*)} \leq_U 
    \tup{L(p[\vec{x}], \vec{y}), C(p)}
\]
But we cannot prove that the following statement is true:
\[
    \tup{L(
    p^*[\vec{x}] - \epsilon, \vec{y}
    ), C(p^*)} \leq_U 
    \tup{L(p[\vec{x}], \vec{y}), C(p)}
\]
And therefore, just using predicates $\cP$, we cannot prove that
$p_r$ is $\epsilon$-correct.
\eat{
\[
    p_r \in \mathds{B}_\epsilon(\mathsf{argmin}_{p \in G} U(L(p[\vec{x}], 
    \vec{y}), C(p)))
\]
}

Therefore, in order to find the optimal program and prove it's optimality,
we have to expand the set of predicates $\cP$. 

To achieve this goal, the algorithm first selects an input-output example from
the noisy dataset $\cD$
on which can improve the difference between the abstract
loss and the concrete loss of programs using procedure $\mathsf{PickDimension}$.
Given an input-output example $(x, y)$,
the idea here is to expand the set of predicates $\cP$ to $\cP'$, such that:
\[
    L(\exec{p^*}^\cP x, y) < L(\exec{p^*}^{\cP'}x, y) \leq L(p^*[x], y)
\]
Thus improving our estimation of the abstract loss function for programs in $G$.

The algorithm allows us to plug any implementation of the procedure 
$\mathsf{PickDimension}$, assuming it satisfies the
following constraint:
\[
\begin{array}{c}
\tup{x_i, y_i} = \mathsf{PickDimension}\big(p, 
    (\vec{x}, \vec{y}), \cP, L\big)
\implies
    L(p[x_i], y_i) > L(\exec{p}^\cP x_i, y_i)
\end{array}
\]

Since $L(p^*[\vec{x}], \vec{y}) - L(\exec{p^*}^\cP\vec{x}, \vec{y}) > \epsilon$
(as $\mathsf{Distance}(p^*, \cD, \cP, L) > \epsilon$),
there exists at least one $i \in [1, n]$, such that,
\[
    L(p^*[x_i], y_i) > L(\exec{p^*}^\cP x_i, y_i)
\]
If multiple input-output examples exist for which the
abstract loss is less than the concrete loss, an implementation of 
$\mathsf{PickDimension}$ can return any one of them and our synthesis algorithm
will use that example to optimize the automaton.

Given the input-output example $(x, y)$, the algorithm uses the procedure
$\mathsf{OptimizeAndBackPropogate}$ to expand the set of predicates to $\cP'$,
such that
\[
    L(\exec{p^*}^\cP x, y) < L(\exec{p^*}^{\cP'}x, y) \leq L(p^*[x], y)
\]

\begin{figure}[tb]
\begin{algorithmic}[1]
    \Procedure{OptimizeAndBackPropagate}{$p, x_j, y_j, \cP, \cU, L$}
    \Statex {\bf input:} 
    Program $p$, input $x_j$, noisy output $y_j$, 
    predicates $\cP$, universe of predicates $\cU$, and Loss Function $L$.
    \Statex {\bf output:} A set of predicates $\cP_r$.
    \State $\cvp := \exec{p}^\cP x_j$;  
    $\phi := (s_0 = \exec{p}x_j)$; 
     \State $\Phi := \big\{ q \in \cU \vert \phi \implies q \big\}$;
     \State $\Psi := \Phi$;
    \For{$i = 1 \ldots m$} \Comment{Use a maximum of $m$ predictates.}
    \State $\Psi := \Psi \bigcup \big\{ 
    \psi \wedge q~\vert~\psi \in \Psi, q \in \Phi
    \big\}$;
    \EndFor
    \State $\psi^* := \phi$;
    \For{$\psi \in \Psi$ }
    \If{$\psi^* \implies \psi$ and 
    $L(\cvp, y_j) - L(\cvp \wedge \psi, y_j) \leq \delta > 0$
    }
$\psi^* := \psi$;
    \EndIf
    \Statex\Comment{The abstract loss is increased by atleast $\delta$.}
    \EndFor
   \State
    {\bf return} 
    $\mathsf{ExtractPredicates}(\psi^*) \bigcup
    %\big[\cvp \wedge \psi^*, 
   \mathsf{BackPropagate}(p, x_j, \cvp \wedge \psi^*, \cP, \cU) $;
\EndProcedure
\end{algorithmic}
\caption{
    Algorithm for extracting predicates $\cP_r$ to refine the abstract value of $p$,
    such that, \\
    $L(\exec{p}^\cP x_j, y_j) < L(\exec{p}^{(\cP \bigcup \cP_r)} x_j, y_j)$.
    %the loss of program $p$ over
    %noisy input-output example $(x_j, y_j)$.
}
\label{alg:constructproof}
\end{figure}

\begin{figure}[tb]
\begin{algorithmic}[1]
    \Procedure{BackPropogate}{$f(e_1, \ldots e_n), 
    x_j, \psi_p, \cP, \cU$}
    \Statex {\bf output:} A set of predicates $P_r$, such that,
    $\exec{f(e_1, \ldots e_n)}^{(\cP \cup \cP_r)}x_j \implies \psi_p$.
    \State 
    $\vec{\phi} := \tup{\exec{e_1} x_j, \ldots \exec{e_n}x_j}$;   
%\State 
    $\vec{\cvp} := \tup{\exec{e_1}^\cP x_j, \ldots \exec{e_n}^\cP x_j}$;
    \State $\vec{\Phi} := \tup{\Phi_1, \ldots \Phi_n}$ where
    $\Phi_i := \big\{ q \in \cU \vert \phi_i \implies q \big\}$;
%\State 
    $\vec{\Psi} := \vec{\Phi}$;
    \For{$i = 1 \ldots m$} \Comment{Use a maximum of $m$ predicates.}
    \For{$j=1, \ldots, n$}
    \State $\Psi_j := \Psi_j \bigcup \big\{ 
    \psi \wedge q~\vert~\psi \in \Psi_j, q \in \Phi_j
    \big\}$;
    \EndFor
    \EndFor
    \State $\vec{\psi^*} := \vec{\phi}$;
    \For{{\bf all } $\vec{\psi} = \tup{\psi_1, \ldots \psi_n}~\vert~\psi_i \in \Psi_i$ }
    \If{$\forall i=1, \ldots n.~\psi_i^* \implies \psi_i$ and 
    $\exec{f(\cvp_1 \wedge \psi_1, \ldots \cvp_n \wedge \psi_n)}^\# \implies
    \psi_p$}
    $\vec{\psi^*} := \vec{\psi}$;
    \EndIf
    \EndFor
    \State $\cP_r := \emptyset$;
    \For{$i = 1 \ldots n$}   
   \State 
    $\cP_r := \cP_r \bigcup \mathsf{ExtractPredicates}(\psi^*_i)$; 
    \If{$e_i \notin T$}
   %$\cI_i := [\cvp_i \wedge \psi^*_i, t]$
%\Else
%\State 
    $\cP_r := \cP_r %\bigcup \mathsf{ExtractPredicates}(\psi^*_i) 
    \bigcup  \mathsf{BackPropogate}(e_i, 
    x_j, \cvp_i \wedge \psi^*_i, \cP, \cU)$;
\EndIf
\EndFor
    \State {\bf return} $\cP_r$;
%\State
  %  {\bf return} 
    %$f({\cI_1, \ldots \cI_n})$
\EndProcedure
\end{algorithmic}
\caption{Algorithm to back propagate abstract value $\cvp \wedge \psi^*$ of expression 
    $e = f(e_1, \ldots e_k)$, such that,\\ $\exec{f(\cvp_1 \wedge \psi^*_1, \ldots
    \cvp_k \wedge \psi^*_k)}^\# \implies \psi_p$.}
\label{alg:backpropogate}
\end{figure}

Figure~\ref{alg:constructproof} presents the $\mathsf{OptimizeAndBackPropogate}$
procedure. The procedure tries to find the strongest formula $\psi^*$, 
such that, $(s_0 = p[x]) \implies \psi^*$ and:
\[
    L((\exec{p^*}^\cP x), y)  <  L((\exec{p^*}^\cP x) \wedge \psi^*, y)  
\]
Note that since $(s_0 = \exec{p^*}x) \implies \psi^*$ (line 7):
\[
    L((\exec{p^*}^\cP x) \wedge \psi^*, y) \leq L(p^*[x], y)
\]

\begin{theorem}\label{thm:backpropogate}

Given expression $e = f(e_1, \ldots e_n)$, input $x$, abstract value $\psi_p$ 
(assuming $(s = \exec{e}x) \implies \psi_p$, 
predicates $\cP$, and universe of predicates $\cU$,
if the procedure $\mathsf{BackPropogate}(e, x, \psi_p, \cP,
\cU)$ returns predicate set $\cP_r$ then: 
\[
    \exec{e}^{\cP \cup \cP_r} x \implies \psi_p
\]
    \eat{
    Let $\cP_r = \mathsf{BackPropogate}(e = f(e_1, \ldots e_k), x, \psi_p, \cP,
    \cU)$. If $\psi_p \implies (s = \exec{e} x)$ then 
    \[
        \exec{e}^{\cP \cup \cP_r} x \implies \psi_p
    \]
}
\end{theorem}
We present the proof of this theorem in the appendix~\ref{sec:appendixproofs}
(Theorem~\ref{thm:appendixbackpropogate}).
\eat{
\begin{proof}
    We prove this theorem using induction over height of expression $e$.
    
    \noindent{\it Base Case:} Height of $e$ is $2$. This means all sub-expressions 
$e_1, \ldots e_k$ are terminals.
    Note that $\cP_r \subseteq \mathsf{ExtractPredicates}(\psi^*_i)$, for all $i
    \in [1, k]$.
    \[
        \exec{e_i}^{\cP \cup \cP_r} \implies \cvp_i \wedge \psi^*_i
    \]
    and
    \[
 \exec{f(\cvp_1 \wedge \psi^*_1, \ldots
    \cvp_k \wedge \psi^*_k)}^\# \implies \psi_p      
    \]
    therefore
    \[
        \exec{e}^{\cP \cup \cP_r} x \implies \psi_p
    \]

    \noindent{\it Induction Hypothesis:} For all expressions $e$ of height less
    than equal to $n$, the following is true:
    \[
        \exec{e}^{\cP \cup \cP_r} x \implies \psi_p
    \]

    \noindent{\it Induction Step:} Let $e = f(e_1, \ldots e_k)$ be an expression
    of height equal to $n+1$.
    The height of expressions $e_1, \ldots e_k$ is less than equal to $n$.

    Note that $\cvp_i \wedge \psi^*_i \implies \exec{e_i}x$ (line-7 and line-9).
    And since $\mathsf{BackPropogate}(e_i,
    x, \cvp_i \wedge \psi^*_i, \cP, \cU) \subseteq \cP_r$, using induction
    hypothesis:
    \[
        \exec{e_i}^{\cP \cup \cP_r} \implies \cvp_i \wedge \psi^*_i
    \]
    and 
    \[
 \exec{f(\cvp_1 \wedge \psi^*_1, \ldots
    \cvp_k \wedge \psi^*_k)}^\# \implies \psi_p      
    \]
    therefore
    \[
        \exec{e}^{\cP \cup \cP_r} x \implies \psi_p
    \]
\end{proof}
}

\begin{theorem}\label{thm:optimize}
    Let $\cP_r = \mathsf{OptimizeAndBackPropogate}(p^*, x, y, \cP, \cU)$.
    \[
        L(\exec{p^*}^\cP x, y) < L(\exec{p^*}x, y) 
    \implies  
    L(\exec{p^*}^\cP x, y) < L(\exec{p^*}^{(\cP \bigcup \cP_r)}x, y) 
    \]
\end{theorem}
\begin{proof} 
    Let $\cvp = \exec{e}^\cP x$ and $\psi^*$ be the abstract value from which
    predicates are extracted (line 10, Figure~\ref{alg:constructproof}).
    If $\psi^*$ was assigned by the $\mathsf{if}$ condition (line 9), then
    \[
        L(\exec{p^*}^\cP x, y) < L(\cvp \wedge \psi^*, y)  
        \]
        However if $\psi^*$ was not assigned on line 9, then $\psi^* = (s_0 =
        p[x])$, and the following is true:
    \[
         L(\exec{p^*}^\cP x, y) < L(\exec{p^*}x, y) = L(\cvp \wedge \psi^*, y)       
    \]
    From Theorem~\ref{thm:backpropogate}, 
    \[
        \exec{p^*}^{\cP \cup \cP_r} \implies \cvp \wedge \psi^* 
    \]
    This implies
    $
        L(\cvp \wedge \psi^*, y) \leq L(\exec{p^*}^{(\cP \bigcup \cP_r)}x, y) 
    $ 
    and therefore
    $
        L(\exec{p^*}^\cP x, y) < L(\exec{p^*}^{(\cP \bigcup \cP_r)}x, y)  
    $.
\end{proof}

\eat{
\begin{theorem}\label{thm:progress}{\bf (Progress)}
    Let $\cA_i$ be the FTA constructed in the $i^{th}$ iteration of
    Algorithm~\ref{alg:top-level}. Let $\cP_i$ be the set of predicates
    and let $p_i$ be the program returned 
    by function $\mathsf{MinCost}$ on the $i^{th}$ iteration.
    Let 
    Let $P_{i, k} \subseteq G$ be the set of programs, such that, $p \in P_{i,
    k}$ if and
    only if 
    \[
        \tup{L(\exec{p}^{\cP_{i+1}}\vec{x}, \vec{y}), C(p)} 
        \leq_U \tup{L(\exec{p_i}^{\cP_i}\vec{x}, \vec{y}), C(p_i)} 
    \]
    If for all $p \in P_{i,i}. \mathsf{Distance}(p, \cD, \cP, L) > \epsilon$
    \[
        P_{i, i+1} \subset P_{i, i} 
    \]
\end{theorem}
\begin{proof}
    The program $p_i \in P_{i, i}$.
    If for all $p \in P_{i,i}. \mathsf{Distance}(p, \cD, \cP, L) > \epsilon$, 
    then $\mathsf{Distance}(p_i, \cD, \cP, L) > \epsilon$.
    Therefore from Theorem~\ref{thm:optimize},
    \[
        L(\exec{p^*}^{\cP_i}\vec{x}, \vec{y}) <
        L(\exec{p^*}^{\cP_{i+1}}\vec{x}, \vec{y})  
    \]
Therefore $P_{i+1} \subseteq P_i - \{p_i\} \subset P_i$.
\end{proof}
}

\begin{theorem}\label{thm:completeness}{\bf (Completeness)}
    Given a Dataset $\cD$, a DSL $G$, tolerance $\epsilon \geq 0$, universe of
    predicates $\cU$, initial predicates $\cP$, objective function $U$, loss
    function $L$, and the complexity measure $C$, 
    the algorithm~\ref{alg:top-level} will eventually return some program $p_r$. 
\end{theorem}
\begin{proof}
    Let $\cA_i$ be the FTA constructed in the $i^{th}$ iteration of
    Algorithm~\ref{alg:top-level}. Let $\cP_i$ be the set of predicates
    and let $p_i$ be the program returned 
    by function $\mathsf{MinCost}$ on the $i^{th}$ iteration.
    From Theorem~\ref{thm:optimize},
    if $\mathsf{Distance}(p_i, \cD, \cP_i, L) > \epsilon)$, then for all $k >
    i$:
    \[
        \mathsf{Distance}(p_i, \cD, \cP_{i+1}, L) = 0 \text{ or } 
        \mathsf{Distance}(p_i, \cD, \cP_i, L) -  
        \mathsf{Distance}(p_i, \cD, \cP_{i+1}, L) \geq \delta > 0 
    \]
    Since we restrict the size of the AFTA (by only considering programs of
    height less than equal to some bound).
    We will reduce the $\mathsf{Distance}$ for some program in $G$ in every
    iteration. 
By induction, we will eventually for some $k$, 
    \[
    \mathsf{Distance}(p_k, \cD, \cP_k, L) \leq \epsilon \] 
\end{proof}

\begin{theorem} {\bf (Correctness)}
    Given a dataset $\cD$, a DSL $G$, tolerance $\epsilon \geq 0$, universe of
    predicates $\cU$, initial predicates $\cP$, objective function $U$, loss
    function $L$, and the complexity measure $C$, 
    the algorithm~\ref{alg:top-level} will return a program $p_r$ which
    satisfies the {\it $\epsilon$-correctness}
    condition~\ref{def:epsiloncorrectness}. 
\end{theorem}
\begin{proof}
    From Theorem~\ref{thm:completeness}, algorithm~\ref{alg:top-level} will
    eventually terminate and return a program $p_r$. From
    Theorem~\ref{thm:soundness}, the returned program $p_r$ will satisfy the
    {\it $\epsilon$-correctness} condition.
\end{proof}

\section{\sys{}~ Implementation}\label{sec:implementation}
We have implemented our synthesis algorithm in a tool
called \sys. \sys{}~ is written in Java. The implementation is modular and allows a user to plug-in
different DSLs, abstract semantics, loss functions, objective functions, and
complexity measures. To support the experiments presented in Section~\ref{sec:results},
we instantiate the \sys{}~ implementation with the string-processing domain-specific
language from\cite{wang2017program, handa2020inductive}.

\noindent{\bf Domain Specific Language and Abstractions:}
We use the string processing domain specific language
from~\cite{wang2017program, handa2020inductive} (Figure~\ref{fig:sygusdsl}),
which supports extracting substrings (using the
$\mathsf{SubStr}$ function) of the input string
$x$ and concatenation of substrings (using the $\mathsf{Concat}$ function).
The function $\mathsf{SubStr}$ function extracts a substring using a start and an end
position. A position can either be a constant index ($\mathsf{ConstPos}$) or 
the start or end of the $k^{th}$ occurrence of the match token $\tau$ 
in the input string ($\mathsf{Pos}$).

\begin{figure}
\[
\begin{array}{rcl}
\text{String expr } e &:=& \mathsf{Str}(f)~\vert~\mathsf{Concat}(f, e);\\
\text{Substring expr } f &:=& \mathsf{ConstStr}(s)~\vert~\mathsf{SubStr}(x, p_1,
p_2);\\
\text{Position } p &:=& \mathsf{Pos}(x, \tau, k,
    d)~\vert~\mathsf{ConstPos}(k);\\
    \text{Direction } d &:=& \mathsf{Start}~\vert~\mathsf{End};
\end{array}
\]
\caption{DSL for string transformations where $\tau$ represents a token, $k$ is
an integer, and $s$ is a string constant.}
\label{fig:sygusdsl}
\end{figure}

\begin{figure}
\[
\begin{array}{rcl}
\exec{f(s_1 = c_1, \ldots, s_k = c_k)}^\# &:=& (s = \exec{f(c_1, \ldots c_k)})\\
\exec{\mathsf{Concat}(\mathsf{len}(f) = i_1,\mathsf{len}(e) = i_2 )}^\# 
&:=& (\mathsf{len}(e) = (i_1 + i_2))\\
\exec{\mathsf{Concat}(\mathsf{len}(f) = i_1, e[i_2] = c )}^\# 
&:=& (e[i_1 + i_2] = c)\\
\exec{\mathsf{Concat}(\mathsf{len}(f) = i, e = c )}^\# 
&:=& (\mathsf{len}(e) = (i + \mathsf{len}(c)) \wedge
 \bigwedge\limits^{\mathsf{len}(c)}_{j=1} e[i + j-1] = c[j-1]\\
\exec{\mathsf{Concat}(f[i] = c, p)}^\# 
&:=& (e[i] = c)\\
\exec{\mathsf{Concat}(f = c,\mathsf{len}(e) = i )}^\# 
&:=& (\mathsf{len}(e) = (\mathsf{len}(c) + i)) \wedge 
\bigwedge\limits^{\mathsf{len}(c)}_{j=1} e[j-1] = c[j-1]\\
\exec{\mathsf{Concat}(f = c_1, e[i] = c_2)}^\# 
&:=& (e[\mathsf{len}(c_1) + i] = c_2) \wedge 
\bigwedge\limits^{\mathsf{len}(c_1)}_{j=1} e[j-1] = c_1[j-1]\\
\exec{\mathsf{Str}(p)}^\# &:=& p
\end{array}
\]
\caption{Abstract Semantics for String Transformation DSL.}
\label{fig:sygusabstract}
\end{figure}

\noindent{\bf Universe of Predicates:}
We construct a universe of predicates 
using predicates of the form $\mathsf{len}(s) = i$, where $s$ is a symbol of a
type of string and $i$ presents an integer. 
We also include predicates of the form $s[i] = c$ indicating the $i^{th}$
character of string $s$ is $c$.
Besides these predicates, we also include predicates of the form $s = c$, where
$c$ is a value which a symbol $s$ can take. We also include both $\mathsf{true}$
and $\mathsf{false}$. In summary, the universe of predicates, we are using, is:
\[
\cU = \big\{\mathsf{len}(s) = i~\vert~i \in \mathds{N}\big\} \cup \big\{
s[i] = c~\vert~i \in \mathds{N}, c \in \mathsf{Char}\big\} \cup
\big\{s = c~\vert~c \in \mathsf{Type}(s)\big\} \cup \big\{\mathsf{true}, \mathsf{false}\big\}
\]

\noindent{\bf Abstract Semantics:}
We define a generic transformer for 
conjunctions of predicates as follows:
\[
    f\bigg((\bigwedge\limits_{i_1}p_{i_1}), \ldots,
    (\bigwedge\limits_{i_k}p_{i_k}) \bigg) := \bigsqcap\limits_{i_1} 
    \ldots \bigsqcap\limits_{i_k} f(p_{i_1}, \ldots p_{i_k})
\]
This allows us to just define an abstract semantics for every possible
combination of atomic predicates, instead of abstract semantics for all possible
abstract values.
Figure~\ref{fig:sygusabstract} presents the abstract semantics for functions in
string processing DSL for all possible combinations of atomic predicates.

\noindent{\bf Initial Abstraction:}
The initial abstraction set $\cP$ includes predicates of form $\mathsf{len}(s) =
i$, where $s$ is a symbol of type string and $i$ is an integer. It also includes
$\mathsf{true}$ and $\mathsf{false}$.

\noindent{\bf Abstractions and Loss Functions:}
We present the abstract version of the $0/\infty$ Loss Function and $0/1$ 
Loss Function below:
\begin{equation*}
    L_{0/\infty}(\cvp, y) = 0 \text{ if } y \in \gamma(\cvp), \infty \text{
        otherwise} 
%\end{equation*}
%\begin{equation*}
    \text{ and }
    L_{0/1}(\cvp, y) = 0 \text{ if } y \in \gamma(\cvp), 1 \text{
        otherwise} 
\end{equation*}
If $\cvp \neq \mathsf{false}$ ($L_{DL}(\mathsf{false}, y) = \infty$), the
abstract version of the Damerau-Levenshtein is 
$L_{DL}(\cvp, y) = d_{c, y}(|c|, |y|)$, where $c = \mathsf{ToStr}(\cvp, y)$
and $d$ is defined below:

\begin{equation*}
    d_{c, y}(i, j) = \min \begin{cases}
        j & i = 0\\
        i & j = 0\\
        d_{c, y}(i-1, j-1) & i, j > 0 \text{ and } (c[i-1] = \mathsf{null} \text{
            or } c[i-1] = y[j-1])\\
        1 + d_{c, y}(i-1, j-1) & i, j > 0 \text{ and } (c[i-1] \neq \mathsf{null} \text{
            and } c[i-1] \neq y[j-1])\\
        1 + d_{c, y}(i-1, j) & i > 0 \\ 
        1 + d_{c, y}(i, j-1) & j > 0 \\
        d_{c, y}(i, j-1) & i = |y| \text{ and } \cvp \text{ may contain strings
        of multiple lengths.}\\
        1 + d_{c, y}(i-2, j-2) & i, j > 1  \text{ and } (c[i-1] = \mathsf{null}
        \text{ or } c[i-1] = y[i-2]) \\ & \text{ and } 
(c[i-2] = \mathsf{null}
        \text{ or } c[i-2] = y[i-1]) 
   \end{cases}
\end{equation*}
Let $\cP = \mathsf{ExtractPredicates}(\cvp)$. 
The procedure $\mathsf{ToStr}$ returns an array $c$, such that, 
if $\mathsf{len}(s) = i \in \cP$ then $|c| = i$, otherwise it is the maximum of
the length of string $y$ or $i$ such that $s[i] = c' \in \cP$.
For all $s[i] = c_i \in \cP$, $c[i] = c_i$, otherwise it is $\mathsf{null}$.

In addition to the loss functions introduced in Section~\ref{sec:preliminaries} 
use the following loss functions:
\begin{definition} {\bf 1-Delete Loss Function:} The 1-Delete loss function 
returns 0 if the outputs from the synthesized program and the data set 
match exactly, 1 if a single
deletion enables the output from the synthesized program to match
the output from the data set, and $\infty$ otherwise: 
\begin{equation*}
    L_{1D}(z, y) = \begin{cases}
        0 & z = y\\
        1 & a \cdot c \cdot b = z \wedge a \cdot b = y \wedge 
\vert c \vert = 1\\ 
        \infty & \text{ otherwise}
    \end{cases}
\end{equation*}
\end{definition}
The abstract version of the 1-Delete Loss Function:
\begin{equation*}
    L_{1D}(\cvp, y) = \begin{cases}
        0
         & y \in \gamma(\cvp)  \\
         1 & a\cdot b = y \text{ and } (\exists c. a \cdot c \cdot b \in
         \gamma(\cvp) \text{ and } |c| = 1 )\\
        \infty & \text{otherwise}
    \end{cases}
\end{equation*}

\begin{definition}
{\bf $n$-Substitution Loss Function:}  The $n$-Substitution loss function 
counts the number of positions where the noisy output
does not agree with the output from the synthesized program.  If the synthesized
program produces an output that is longer or shorter than the output in the 
noisy data set, the loss function is $\infty$:
    \begin{equation*}
    L_{nS}(z, y) = \begin{cases}
        \infty & \vert z \vert \neq \vert y \vert \\
        \sum\limits_{i=1}^{\vert z \vert} 1 \mbox{ if } z[i] \neq y[i] \mbox{ else } 0
    & \vert z \vert = \vert y \vert 
    \end{cases}
    \end{equation*}
\end{definition}
The abstract version of the $n$-Substitution loss function is:
\begin{equation*}
    L_{nS}(\cvp, y) = \begin{cases}
        0 & \cvp = \mathsf{true}\\
        \infty & \cvp = \mathsf{false}\\
        L_{nS}(c, y) & \cvp = (s = c)\\
         \infty  & c_\cvp = \mathsf{ToStr}(\cvp, y)
        \text{ and } |c_\cvp| \neq |y|
         \\
        \sum\limits_{j=1}^{|y|} \mathds{1}(c[j] \neq \mathsf{null} \text{ and }c[j] \neq y[i_j]) 
        & c = \mathsf{ToStr}(\cvp, y)
        \text{ and } |c| = |y|
   \end{cases}
\end{equation*}

\noindent{\bf Incremental Automata Update:} To avoid regenerating the entire
FTA at every iteration of the algorithm as in Algorithm~\ref{alg:top-level}, our
\sys{}~ implementation applies on optimization that incrementally updates parts
of the FTA as appropriate at every iteration of the algorithm. 

\section{Experimental Results}
\label{sec:results}

We use the SyGuS 2018 benchmark suite~\cite{alur2013syntax} to evaluate 
\sys~ against the current state-of-the-art noisy program synthesis
system presented in \cite{handa2020inductive}.
The SyGus 2018 benchmark suite contains a range of string transformation problems,
a class of problems that has been extensively studied in past
program synthesis projects~\cite{gulwani2011automating, polozov2015flashmeta, singh2016transforming}.
\cite{handa2020inductive} use this benchmark suite to benchmark their
system by systematically introducing noise within these benchmarks. We recreated
the scenarios studied in~\cite{handa2020inductive} and report results for 
these scenarios. 

We run all experiments on a 3.00 GHz Intel(R) Xeon(R) CPU E5-2690 v2 machine
with 512GB memory running Linux 4.15.0. We set a timeout limit of 10 minutes for
each synthesis task. We compare \sys{}~ with the noisy program 
synthesis system presented in~\cite{handa2020inductive} 
(see Section~\ref{sec:preliminaries}),
running with a {\it bounded scope height threshold of $4$} for all experiments
(\cite{handa2020inductive}, Section 9.1). We call this system CFTA.

\subsection{Noisy Data Sets}

% Table generated by Excel2LaTeX from sheet 'Sheet3'
\begin{figure}
    {\small
    \centering
    \begin{tabular}{|l|c|c|c|c|c|c|c|c|}
    \hline
        \multirow{3}[0]{*}{Benchmark} & \multirow{3}[0]{*}{No of examples } &
        \multicolumn{6}{c|}{Rose}                      & \multirow{3}[0]{*}{CFTA} \\
          \cline{3-8}
          &       & \multicolumn{3}{c|}{$n=1$} & \multicolumn{3}{c|}{$n=3$} &  \\
          \cline{3-8}
          &       & $L_{0/1}$   & $L_{DL}$    & $L_{1D}$    & $L_{0/1}$   & $L_{DL}$   & $L_{1D}$   &  \\
    \hline
    bikes & 6     & 68    & 67    & 65    & 67    & 67    & 68    & 19554 \\
    bikes\_small & 6     & 69    & 68    & 65    & 67    & 65    & 65    & 21210 \\
    dr-name & 4     & 402   & 310   & 408   & 359   & 324   & 390   & - \\
    dr-name\_small & 4     & 424   & 312   & 376   & 330   & 302   & 375   & - \\
    firstname & 4     & 118   & 105   & 114   & 96    & 319   & 367   & 4258 \\
    firstname\_small & 4     & 113   & 106   & 118   & 96    & 300   & 393   & 4220 \\
    initials & 4     & 336   & 289   & 353   & 244   & -     & -     & 36188 \\
    initials\_small & 4     & 361   & 293   & 343   & 249   & -     & -     & 30920 \\
    lastname & 4     & 124   & 113   & 120   & 97    & 114   & 121   & 175762 \\
    lastname\_small & 4     & 122   & 114   & 120   & 97    & 116   & 121   & 178825 \\
    name-combine & 6     & 1288  & 918   & 1609  & 1330  & 984   & 1477  & - \\
    name-combine\_short & 6     & 1301  & 875   & 1528  & 1333  & 945   & 1437  & - \\
    name-combine-2 & 4     & 2100  & 1721  & 2162  & 740   & -     & -     & - \\
    name-combine-2\_short & 4     & 2120  & 1697  & 2188  & 735   & -     & -     & - \\
    name-combine-3 & 6     & 298   & 287   & 294   & 301   & 289   & 299   & 547447 \\
    name-combine-3\_short & 6     & 313   & 294   & 291   & 296   & 278   & 290   & 544044 \\
    name-combine-4 & 5     & 1863  & 1683  & 1917  & 1875  & 1584  & 1921  & - \\
    name-combine-4\_short & 5     & 1888  & 1664  & 1914  & 1815  & 1604  & 1875  & - \\
    phone & 6     & 66    & 68    & 67    & 66    & 68    & 64    & 943 \\
    phone\_short & 6     & 65    & 68    & 66    & 67    & 66    & 68    & 963 \\
    phone-1 & 6     & 62    & 62    & 61    & 59    & 63    & 62    & 933 \\
    phone-1\_short & 6     & 61    & 64    & 60    & 62    & 63    & 60    & 942 \\
    phone-2 & 6     & 83    & 73    & 79    & 80    & 79    & 78    & 953 \\
    phone-2\_short & 6     & 117   & 76    & 80    & 80    & 76    & 79    & 943 \\
    phone-3 & 7     & 926   & 559   & 692   & 839   & 611   & 648   & - \\
    phone-3\_short & 7     & 789   & 530   & 652   & 786   & 607   & 600   & - \\
    phone-4 & 6     & 2571  & 2067  & 3054  & 2678  & 2202  & 2963  & - \\
    phone-4\_short & 6     & 2637  & 2057  & 2872  & 2608  & 2256  & 3014  & - \\
    phone-5 & 7     & 114   & 101   & 110   & 112   & 98    & 108   & 122 \\
    phone-5\_short & 7     & 109   & 99    & 105   & 116   & 98    & 110   & 127 \\
    phone-6 & 7     & 171   & 138   & 168   & 175   & 142   & 166   & 3230 \\
    phone-6\_short & 7     & 170   & 140   & 169   & 173   & 135   & 163   & 3327 \\
    phone-7 & 7     & 165   & 132   & 153   & 159   & 134   & 155   & 2793 \\
    phone-7\_short & 7     & 165   & 133   & 162   & 157   & 136   & 154   & 2762 \\
    phone-8 & 7     & 158   & 144   & 157   & 162   & 145   & 156   & 3464 \\
    phone-8\_short & 7     & 157   & 142   & 154   & 162   & 141   & 155   & 3223 \\
    phone-9 & 7     & 28815 & 29672 & 28029 & 28576 & 26465 & 30941 & - \\
    phone-9\_short & 7     & 28658 & 29055 & 28729 & 29115 & 27818 & 28157 & - \\
    phone-10 & 7     & 87772 & 73379 & 66170 & 76963 & 87778 & 69130 & - \\
    phone-10\_short & 7     & 85861 & 77495 & 75849 & 86508 & 81526 & 65247 & - \\
    reverse-name & 6     & 699   & 645   & 697   & 703   & 666   & 671   & - \\
    reverse-name\_short & 6     & 709   & 662   & 811   & 712   & 625   & 685   & - \\
    univ\_3 & 6     & 6258  & 4117  & 5994  & 6260  & 3499  & 6068  & - \\
    univ\_3\_short & 6     & 6345  & 4364  & 6503  & 6331  & 3510  & 5962  & - \\
        \hline
    \end{tabular}%
    }
\caption{Runtime performance of \sys~and CFTA on benchmarks with deletion based
noise.}
    \label{table:deletions}%
\end{figure}%

Table~\ref{table:deletions} presents results for all SyGus 2018 benchmark
problems which contain less than ten input/output examples.
We omit univ\_1, univ\_2, and univ\_4-6 --- these 
problems time out for both \sys{}~ and CFTA (so the rows would contain
all -). The first column (Benchmark) presents the name of the SyGus 2018
benchmark. The second column (Number of Input/Output Examples) presents the
number of input/output examples in the benchmark problem. The remaining
columns present running times, in milliseconds, for \sys{}~ and CFTA running
with different noise sources and loss functions. 
A - indicates that the corresponding run timed
out without synthesizing a program. 

The objective function is the lexicographic
objective function. The complexity measure is program size. The noise source
cyclically deletes a single character from outputs in the data set,
starting with the first character, then wrapping around when reaching the
last position in the output. When $n=1$, the noise source corrupts the last
output in the set of input/output examples. When $n=3$, the noise source
corrupts the last $3$ input/output examples in the set of input/output examples.
For \sys{}~, the table presents results for each of the $L_{0/1}, L_{DL},$ and $L_{1D}$
loss functions. For CFTA, we report one running time for each benchmark problem --- 
for CFTA, the running time is the same for all noise source/loss function combinations. 

For the benchmarks on which both terminate, \sys{}~ runs up to 1957 times
 faster than CFTA, with a median speedup of 20.5 times over CFTA. 
This performance increase enables \sys{}~ to successfully synthesize programs for 
20 more benchmark problems than CFTA --- \sys{}~ synthesizes programs for 
44 of the 54 benchmark problems (timing out on the remaining 10), while
CFTA can only synthesize programs for 24 of the 54 benchmark problems
(timing out on the remaining 30). These results highlight
the substantial performance benefits that \sys{}~ delivers. 

Every synthesized program is guaranteed to minimize the objective function over the given
input/output examples. For $n=1$, all synthesized programs also have zero loss over
the original (unseen during synthesis) noise-free input/output examples (i.e., all synthesized 
programs generate the correct output for each given input). For $n=3$, 
34, 40, and 40 out of 44
synthesized programs, for $L_{0/1}$, $L_{DL}$, and $L_{1D}$ respectively, 
have zero loss over the original noise-free input/output
examples. For a given noise source/loss function combination, CFTA and \sys{}~
synthesize the same program (unless one or both of the systems times out).  
These results highlight the ability of \sys{}~ to synthesize correct programs 
even in the face of significant noise. 

% Table generated by Excel2LaTeX from sheet
% Table generated by Excel2LaTeX from sheet 'Sheet2'
\begin{figure}
    \centering
    \begin{tabular}{|l|c|c|c|c|}
\hline
    \multirow{2}[0]{*}{Benchmark} & \multirow{2}[0]{*}{No of Examples} &
\multicolumn{2}{c|}{Rose} & \multirow{2}[0]{*}{CFTA} \\
\cline{3-4}
          &       & $L_{nS}$    & $L_{DL}$    &  \\
\hline    
    phone-long-repeat & 400   & 745   & 749   & 28836 \\
    phone-1-long-repeat & 400   & 728   & 749   & 28547 \\
    phone-2-long-repeat & 400   & 806   & 784   & 29766 \\
    phone-3-long-repeat & 400   & 1919  & 2297  & - \\
    phone-4-long-repeat & 400   & 3332  & 4683  & - \\
    phone-5-long-repeat & 400   & 1145  & 1156  & 2179 \\
    phone-6-long-repeat & 400   & 1227  & 1315  & 100241 \\
    phone-7-long-repeat & 400   & 1253  & 1341  & 89535 \\
    phone-8-long-repeat & 400   & 1207  & 1290  & 88932 \\
    phone-9-long-repeat & 400   & 18795 & 39032 & - \\
    phone-10-long-repeat & 400   & 41485 & 113397 & - \\
\hline    
\end{tabular}%
  \caption{\sys~and CFTA's performance on dataset corrupted by substitution
based noise.}
\label{table:substitutions}%
\end{figure}%

Table~\ref{table:substitutions} presents results for the SyGus 2018 phone-*-long-repeat 
benchmarks running with a noise source that cyclically and
probabilistically replaces a single digit in each output string with the next
digit (wrapping back to $0$ if the current digit is $9$).
The noise source iterates through each output string in the data set in turn,
probabilistically replacing the next character position in each output string 
with another character, wrapping around to the first character position when
it reaches the last character position in the output string. 
The noise source corrupts $95\%$ of the input-output examples in each dataset.

We report results for two loss functions, $L_{nS}$ ($n$-substitution) and 
$L_{DL}$ (Damerau-Levenshtien). The objective function is the lexicographic
objective function. The complexity measure is program size. There
is a row in the table for each phone-*-long-repeat benchmark; each entry 
presents the running time (in milliseconds) for the corresponding synthesis
algorithm running on the corresponding benchmark problem.  

For the benchmarks on which both terminate, \sys{}~ runs up to 81.7 
times faster than CFTA, with a median speedup of 39.0 times over CFTA. 
This performance increase enables \sys{}~ to successfully synthesize programs for 
4 more benchmark problems than CFTA --- \sys{}~ synthesizes programs for 
all 11 of the benchmark problems, 
while CFTA synthesizes programs for 7 of the 11 benchmark problems
(timing out on the remaining 4). Once again, these results highlight
the substantial performance benefits that \sys{}~ delivers. 

Every synthesized program is guaranteed to minimize the objective function over the given
input/output examples. All synthesized programs have zero loss over the 
original (unseen during synthesis) noise-free input/output examples (i.e., all synthesized 
programs generate the correct output for each given input). 
Once again, these results highlight the ability of \sys{}~ to synthesize correct programs 
even in the face of significant noise. 

\noindent{\bf Noise-Free Data Sets}
We also evaluated the performance of \sys~and CFTA by applying it to all
problems in the SyGuS 2018 benchmark
suite~\cite{alur2013syntax}. For each problem we synthesize the optimal program
over clean (noise-free) datasets.
We present the result of all SyGuS benchmark suite in the
appendix~\ref{sec:appendixresults}. 

\eat{
For \sys~, the table presents results for
$5$ loss functions, zero-one loss function $L_{0/1}$, zero-infinity loss
function $L_{0/\infty}$, Damerau-Levenshtein loss function $L_{DL}$, $1$-Delete 
loss function $L_{1D}$, and the $n$-Substitution loss function $L_{nS}$.
}

For the benchmarks on which both terminate, \sys~runs up to 1831 times faster
than CFTA, with the median speedup of 35.7 over CFTA. This enables \sys~to
successfully synthesize programs for 45 more benchmark problems that CFTA -- 
\sys~synthesizes programs for 90 of the 108 benchmark problems (timing out on
the remaining 18 problems), while CFTA can
only synthesize the correct program for 45 of the 108 benchmark problems (timing
out of the remaining problems).

\vspace{-0.05in}
\section{Related Work}
\vspace{-0.05in}
We discuss related work in the following areas:

\noindent{\bf Programming-By-Example/Noise-Free Synthesis:}
Synthesizing programs from a set of input-output examples has been a prominent
topic of research for many 
years~\cite{shaw1975inferring,gulwani2011automating, singh2016transforming}. 
These techniques either require the entire dataset to be noise free, or they try
to remove corrupted input-output examples from the dataset before synthesizing
the correct program. Instead of removing corrupted examples from the dataset,
our technique uses the loss function to capture information from 
any corrupted examples and uses this information during the synthesis. 
The experimental results show that, for our set of benchmarks, our
approach can synthesize the correct program even in the presence
of substantial noise.
% in the provided input/output examples. 

\noindent {\bf Neural Program Synthesis/Machine-Learning Approaches:}
Researchers have investigated techniques that use machine
learning/deep neural networks to synthesize programs~\cite{raychev2016learning, 
devlin2017robustfill, balog2016deepcoder}. 
The techniques primarily focus on synthesizing programs over noise-free datasets.
These techniques require a training phase and a differentable loss function and
provide no guarantees that the synthesized program will minimize the 
objective function.  Our technique, in contrast, does not require a 
training phase, can work with arbitrary loss functions including, for example, the 
Damerau-Levenshtien loss function, and comes with a guarantee that the
synthesized program will minimize the objective function over the provided
(noisy) input/output examples. 

\noindent{\bf Tree Automata/VSA based synthesis algorithms:}
\cite{singh2016transforming,polozov2015flashmeta, wang2017synthesis} 
all construct either version spaces or tree 
automata to represent all programs (within a bounded search space) which 
satisfy a set of input-output examples.
These techniques build the version space in one shot before synthesis and require the dataset to
be either be correct or pruned to remove any corrupt input/output examples. 

Our technique also works with version spaces (as represented in tree automata),
but extends the approach to work with noisy data sets. It also introduces an abstraction based
optimization process which can iteratively expand and improve the
version space during synthesis (instead of building it in one shot as in previous research). 

\noindent{\bf Abstraction-Refinement based Synthesis Algorithms:}
There has been work  done on using abstraction refinement/refinement types to
synthesize programs~\cite{guo2019program, wang2017synthesis, polikarpova2016program}. 
Given a noise free dataset and a program, checking if a program is correct or
incorrect simply checks if the synthesized programs satisfy all input/output examples. 
To refine an abstraction, these techniques construct a proof of
incorrectness. Each abstraction identifies a set of programs, some of which may
be correct and others of which may be incorrect. Refinement first identifies
a program that does not satisfy one or more of the input/output examples, then
generates constraints that refine the abstraction to eliminate this program. 
Iterative refinement eventually produces the final program. 

Abstraction in our noisy program synthesis framework, in contrast, works with an
abstraction that approximates the loss function over a set of programs. The refinement
step selects a program within the abstraction space, computes its loss, then uses this
computed loss to refine the loss approximation to bring this approximation closer
to the actual loss. This refinement step, in expectation, reduces the inaccuracy
in the approximated loss function of the programs identified by the abstraction. 
In contrast to previous approaches, which work with abstractions based on program
correctness and refinement steps that eliminate incorrect programs, our approach
works with abstractions that maintain a sound, conservative approximation of the minimum loss
function over the set of programs identified by the abstraction and refinement
steps that eliminate programs based on the loss of the programs. 

One key difference is that refinement steps in previous techniques rely on the 
ability to identify correct and incorrect programs. Because our
technique works with noisy data sets, it can never tell if a candidate program has 
minimal loss without comparing the program to all other current candidate
programs (unless the loss happens to be zero). It instead uses abstract
minimum loss values to bound how far off the optimal loss any candidate program
may be. Instead of working with correct or incorrect programs, our technique
works by iteratively improving the accuracy of the minimum 
loss function estimation captured by the abstraction. 

Our technique therefore combines abstract tree automata with an abstraction-based optimization process. 
Our approach, in contrast to previous approaches that use abstract tree automata, 
enables us to synthesize programs that optimize an objective function over 
a set of noisy input/output examples, including synthesizing correct programs that may
disagree with one, some, or even all of the provided input/output examples. 

\cite{wang2017program} uses abstract tree automata
and abstraction refinement for program synthesis. 
Because their refinement strategy prunes any program that
does not satisfy all of the provided input/output examples, 
their algorithm requires the dataset to be noise free. This pruning is necessary as this 
allows their technique to effectively capture constraints to prune large part of the search space. 

\noindent{\bf Noisy Program Synthesis:}
Handa and Rinard formalize a noisy program synthesis framework and 
present a technique that uses finite tree automata to 
synthesize the program which {\it best fits} the noisy dataset based on an 
objective function, a loss function, and a complexity measure~\cite{handa2020inductive}. 
We present a simplified version of the presented technique in Section~\ref{sec:preliminaries}. 
This paper presents a new abstraction/refinement 
technique to solve the noisy program synthesis problem.
Our experimental results show that this new technique significantly outperforms
the technique presented in~\cite{handa2020inductive}.

Handa and Rinard also formalize a connection,
in the context of noisy program synthesis, between the characteristics
of the noise source and the hidden program that together generate the noisy data set and
the characteristics of the loss function~\cite{handa2021program}. Specifically, 
the paper identifies, given a noise model, a corresponding optimal loss
function as well as properties of combinations of noise models and loss functions 
that ensure that the presented noisy synthesis algorithm will converge to the correct 
program given enough noisy input/output pairs. This work is complementary to the
work presented in this paper.  All of the guarantees established 
in ~\cite{handa2021program} also apply to the algorithm presented in this paper. 

\noindent{\bf Best Effort Synthesis:}
\cite{peleg2020perfect} presents an enumeration-based technique to
synthesize programs from input-output datasets containing some incorrect outputs.
Their technique returns a ranked list of partially valid programs, removing programs
which are observationally equivalent. Their technique
uses a fixed fitness function to order these partial results. 
\cite{peleg2020perfect} uses a  specific loss function and a specific 
complexity measure to rank candidate programs. Given this loss function and
complexity measure, our technique will synthesize the exact same program.
Our technique also supports the 
use of a large class of loss functions, complexity measures, and objective functions. 
\cite{handa2020inductive} has showcased how crafting suitable loss functions
is essential and allows one to synthesize the 
correct program even when some or even all input/output examples are corrupted.

\section{Conclusion}
We present a new technique to synthesize programs over noisy datasets. 
This technique uses an abstraction refinement based optimization process to search
for a program which {\it best-fits} a given dataset, based on an objective
function, a loss function, and a complexity measure. The algorithm deploys
an abstract semantics to soundly approximate the minimum loss function over
abstracted sets of programs in an underlying domain-specific language. 
Iterative refinement based on this sound approximation produces a 
program whose loss value is within a specified tolerance level of the 
program with optimal loss over the given noisy input/output data set. 
We provide a proof that the technique is sound and complete and will
always synthesize an $\epsilon$-correct program. 

We have implemented our synthesis algorithm in the \sys{} noisy program
synthesis system. Our experimental results show that,
on two noisy benchmark program synthesis problem sets drawn from the
SyGus 2018 benchmarks, \sys{} delivers speedups of up to 1587 and
81.7 over a previous state-of-the art noisy synthesis system, with 
median speedups of 20.5 and 81.7 over this previous system. \sys{}~ also
terminates on 20 (out of 54) and 4 (out of 11) more benchmark problems
than the previous system.  Both \sys{}~ and the previous system synthesize
programs that are optimal over the provided noisy data sets.  
For the majority of the problems in the benchmark
sets ($272$ out of $286$ for \sys), both systems also synthesize programs
that produce correct outputs for all inputs in the original (unseen) noise-free
data set. These results highlight the significant benefits that \sys{}~ can 
deliver for effective noisy program synthesis.

%% Acknowledgments
\eat{
\begin{acks}                            %% acks environment is optional
                                        %% contents suppressed with 'anonymous'
  %% Commands \grantsponsor{<sponsorID>}{<name>}{<url>} and
  %% \grantnum[<url>]{<sponsorID>}{<number>} should be used to
  %% acknowledge financial support and will be used by metadata
  %% extraction tools.
  This material is based upon work supported by the
  \grantsponsor{GS100000001}{National Science
    Foundation}{http://dx.doi.org/10.13039/100000001} under Grant
  No.~\grantnum{GS100000001}{nnnnnnn} and Grant
  No.~\grantnum{GS100000001}{mmmmmmm}.  Any opinions, findings, and
  conclusions or recommendations expressed in this material are those
  of the author and do not necessarily reflect the views of the
  National Science Foundation.
\end{acks}
}

%% Bibliography
%\bibliography{bibfile}
\bibliography{citation}

\newpage
%% Appendix
\appendix
\section{Appendix}
\subsection{Appendix: Additional Theorems}
\label{sec:appendixproofs}

\begin{theorem}\label{thm:alphap}
    Given $\cP$ and $\cP^*$, such that, $\{\mathsf{true}, \mathsf{false}\} \subseteq \cP 
    \subseteq \cP^* \subseteq \cU$ is true, then for any
    abstract value $\cvp$, the following statement is true:
    \[
        \alpha^{\cP^*}(\cvp) \implies \alpha^{\cP}(\cvp)
    \]
\end{theorem}
\begin{proof}
    Proof by contradiction. Assuming $\alpha^{\cP^*}(\cvp) = \cvp_1$ 
    and $\alpha^{\cP}(\cvp) = \cvp_2$, such that, $\cvp_1 \centernot\implies \cvp_2$.
    Note that:
    \[
        \cvp \implies \cvp_1 \wedge \cvp_2 \text{, } \cvp_1 \wedge \cvp_2
        \implies \cvp_1 \text{, and }
        \cvp_1 \wedge \cvp_2
        \implies \cvp_2 
    \]
    $\cvp_1 \wedge \cvp_2$ can be expressed by using predicates in $\cP^*$,
    and $\cvp_1 \wedge \cvp_2$ is stronger than $\cvp_1$. 
    Hence $\cvp_1 \neq \alpha^{\cP^*}(\cvp)$. 
    Therefore, by contradiction, the above theorem is true.
\end{proof}

\begin{theorem}\label{thm:execp}
Given a set of predicates $\cP \subseteq \cU$ and $\cP^* \subseteq \cU$, such
    that $\{\mathsf{true}\} \subseteq \cP \subseteq \cP^*$, then for any 
    expression $e$ (starting from symbol $s$) and any input value
    $x_i$, the following statement is true:
    \[
        (s = \exec{e} x_i) \implies \exec{e}^{\cP^*} x_i 
        \text{ and } \exec{e}^{\cP^*} x_i \implies \exec{e}^{\cP} x_i
    \]
    i.e., adding more predicates will make the abstract computation more precise.
\end{theorem}
\begin{proof}
    We prove this using induction over height of expression $e$.

    \noindent{\it Base Case:} The height of $e$ is $1$, i.e., $e$ is a terminal
    $t$.
    Since $\mathsf{true} \in \cP$, using definition of $\alpha^\cP$:
    \[
        (s = \exec{t} x_i) \implies \alpha^{\cP^*}(s = \exec{t} x_i)
    \]
    From Theorem~\ref{thm:alphap}:
    \[
 \alpha^{\cP^*}(s = \exec{t} x_i) \implies \alpha^{\cP}(s = \exec{t} x_i)
    \]

    \noindent{\it Induction Hypothesis:} For all expressions $e$ of height less than
    equal to $n$, the following statement is true:
        \[(s = \exec{e} x_i) \implies \exec{e}^{\cP^*} x_i \text{ and } 
        \exec{e}^{\cP^*} x_i  \implies \exec{e}^{\cP}
        x_i\]

    \noindent{\it Induction Step:} Consider an expression $e$ of height $n+1$.
    Without loss of generality, we can assume $e$ is of the form $f(e_1, \ldots
    e_k)$, where height of sub-expressions $e_1, \ldots e_k$ is less than equal
    to $n$.
    Therefore, for all $j \in [1, n]$
    \[
        (s = \exec{e_j} x_i) \implies \exec{e_j}^{\cP^*} x_i
    \text{ and }
\exec{e_j}^{\cP^*} x_i 
        \implies \exec{e_j}^{\cP} x_i
    \]
    and for all $j \in [1, n]$
    \[
        \exec{e_j}x_i \in \gamma(\exec{e_j}^{\cP^*} x_i) \subseteq \gamma(\exec{e_j}^{\cP^*} x_i)   
    \]
    Let $\alpha^\cP(\exec{f(\exec{e_1}^{\cP^*} x_i,\ldots \exec{e_k}^{\cP^*} x_i)}^\#) = \cvp^*$
    and $\alpha^\cP(\exec{f(\exec{e_1}^{\cP} x_i, \ldots \exec{e_k}^{\cP}
    x_i)}^\#) = \cvp$. Note that:
    \[
        \exec{f(\exec{e_1}x_i, \ldots \exec{e_k}x_i)} \in \gamma(\cvp^*) \subseteq \gamma(\cvp)
    \]
    Therefore,
    \[
        (s = \exec{e}x_i) \implies \cvp^* \text{ and } \cvp^* \implies \cvp
    \]
    \[
         (s = \exec{e}x_i) \implies \alpha^{\cP^*}(\cvp^*) 
    \text{ and }
         \alpha^{\cP^*}(\cvp^*) \implies \alpha^{\cP^*}(\cvp) 
    \text{ and }
          \alpha^{\cP^*}(\cvp) \implies \alpha^\cP(\cvp) 
    \]
    Hence, by induction, the above theorem is true.
\end{proof}

\begin{theorem}\label{thm:appendixstructure}{\bf (Structure of the Tree Automaton)}
    Given a set of predicates $\cP$, input vector $\vec{x} = \tup{x_1, \ldots
    x_n}$, and DSL $G$,
    let $\cA = (Q, Q_f, \Delta)$ be the AFTA returned by the
    function $\mathsf{ConstructAFTA}(\vec{x}, G, \cP)$.
    Then for all symbols $s$ in $G$,
    for all expressions $e$ starting from symbol $s$
    (and height less than bound $b$),
    there exists a state $q^{\cVvp}_s \in Q$,
    such that, $e$ is accepted by the automaton $(Q, \{q^\cVvp_s\}, \Delta)$,
    where $\cVvp = \tup{\exec{e}^\cP x_1, \ldots \exec{e}^\cP x_n }$.
%    For any
%    state $q^{\cVvp}_{s} \in Q$, a subexpression $e$ is accepted
%    by the automaton $(Q, \{q^{\cVvp}_{s}\}, \Delta)$
%    if and only if $\forall i \in [1, n].(s = \exec{e}x_i) \implies \cvp_i$.
\end{theorem}
\begin{proof}
    We prove this theorem by using induction over height of the expression $e$.

    \noindent{\it Base Case:} Height of expression $e$ is $1$. This implies the
    symbol is either $x$ or a constant. According to Var and Const rules
    (Figure~\ref{fig:rulesafta}),
    there exists state $q^{\cVvp}_t \in Q$ (for terminal $t$),
    where $\cVvp =
\tup{\exec{t}^\cP x_1, \ldots \exec{t}^\cP x_n }$ and $t$ is accepted by
    automaton $(Q, \{q^\cVvp_t\}, \Delta)$.

    \noindent{\it Inductive Hypothesis:} For all symbols $s$ in $G$, for all
    expressions $e$ starting from symbol $s$ of height less than equal to $n$,
    there exists a state $q^{\cVvp}_s \in Q$,
    such that, $e$ is accepted by the automaton $(Q, \{q^\cVvp_s\}, \Delta)$,
    where $\cVvp = \tup{\exec{e}^\cP x_1, \ldots \exec{e}^\cP x_n }$.

    \noindent{\it Induction Step:} For any symbol $s$ in $G$, consider an
    expression $e = f(e_1, \ldots e_k)$ of height equal to $n+1$,
    created from production $s \leftarrow f(s_1, \ldots s_k)$.
    Note the height of expressions $e_1, \ldots e_k$ is less than equal to $n$,
    therefore using induction hypothesis, there exists states
    $q^{\cVvp_1}_{s_1}, \ldots q^{\cVvp_k}_{s_k} \in Q$, such that
    $e_i$ is accepted by automaton $(Q, \{q^{\cVvp_i}_{s_i}\}, \Delta)$,
    where $\cVvp_i = \tup{\exec{e_i}^\cP x_1, \ldots \exec{e_i}^\cP x_n}$.
    Note based on abstract execution rules (Figure~\ref{fig:execabstract}):
    \[
        \exec{e}^\cP x_i = \alpha^\cP(\exec{f(\exec{e_1}^\cP x_i, \ldots
        \exec{e_k}^\cP x_i )}^\#)
    \]
    According to Prod rule (Figure~\ref{fig:rulesafta}),
    there exists a state $q^\cVvp_s \in
    Q$, where $\cVvp = \tup{\exec{e}^\cP x_1, \ldots \exec{e}^\cP x_n }$, and
    $e$ is accepted by $(Q, \{q^{\cVvp}_s\}, \Delta)$.

    Therefore, by induction, for all symbols $s$ in $G$,
    for all expressions $e$ starting from symbol $s$
    (and height less than bound $b$),
    there exists a state $q^{\cVvp}_s \in Q$,
    such that, $e$ is accepted by the automaton $(Q, \{q^\cVvp_s\}, \Delta)$,
    where $\cVvp = \tup{\exec{e}^\cP x_1, \ldots \exec{e}^\cP x_n }$.
\end{proof}

\begin{theorem}\label{thm:appendixargmin}
    Given predicates $\cP$, DSL $G$, noisy dataset $\cD = (\vec{x}, \vec{y})$, 
    objective function $U$, loss
    function $L$, complexity measure $C$, and $\cA =
    \mathsf{ConstructAFTA}(\vec{x}, G, \cP)$, if $p^* =
    \mathsf{MinCost}(\cA, \cD, U, L, C)$ then 
    \[
        p^* \in \mathsf{argmin}_{p \in G} U(L(\exec{p}^\cP\vec{x}, \vec{y}),
        C(p))
    \]
    i.e., $p^*$ minimizes the {\it abstract} objective function.
\end{theorem}
\begin{proof}
    Let $\cA = (Q, Q_f, \Delta)$.
    From corollary~\ref{cor:abstractvalue}, for each program $p \in G$, there
    exists a state $q^\cVvp_{s_0} \in Q_f$, such that, $\exec{p}^\cP \vec{x} =
    \cVvp$.
Since the algorithm finds an accepting state $q^{\cVvp^*}_{s_0} \in Q_f$, 
such that,
for all accepting states $q^{\cVvp}_{s_0} \in Q_f$,
\[
   \tup{L(\cVvp^*, \vec{y}), C(P[q^{\cVvp^*}_{s_0}])}
\leq_U
   \tup{L(\cVvp, \vec{y}), C(P[q^{\cVvp}_{s_0}])}
\]
for all $p \in G$,
    \[
   \tup{L(\cVvp^*, \vec{y}), C(P[q^{\cVvp^*}_{s_0}])}
\leq_U
    \tup{L(\exec{p}^\cP\vec{x}, \vec{y}), C(p)}
\]
    Since $p^* = P[q^{\cVvp^*}_{s_0}]$,
     \[
        p^* \in \mathsf{argmin}_{p \in G} U(L(\exec{p}^\cP\vec{x}, \vec{y}),
        C(p))
    \]
\end{proof}

\begin{theorem}\label{thm:appendixbackpropogate}
Given expression $e = f(e_1, \ldots e_n)$, input $x$, abstract value $\psi_p$
(assuming $(s = \exec{e}x) \implies \psi_p$,
predicates $\cP$, and universe of predicates $\cU$,
if the procedure $\mathsf{BackPropogate}(e, x, \psi_p, \cP,
\cU)$ returns predicate set $\cP_r$ then:
\[
    \exec{e}^{\cP \cup \cP_r} x \implies \psi_p
\]
    \eat{
    Let $\cP_r = \mathsf{BackPropogate}(e = f(e_1, \ldots e_k), x, \psi_p, \cP,
    \cU)$. If $\psi_p \implies (s = \exec{e} x)$ then
    \[
        \exec{e}^{\cP \cup \cP_r} x \implies \psi_p
    \]
}
\end{theorem}
\begin{proof}
    We prove this theorem using induction over height of expression $e$.

    \noindent{\it Base Case:} Height of $e$ is $2$. This means all sub-expressions
$e_1, \ldots e_k$ are terminals.
    Note that $\cP_r \subseteq \mathsf{ExtractPredicates}(\psi^*_i)$, for all $i
    \in [1, k]$.
    \[
        \exec{e_i}^{\cP \cup \cP_r} \implies \cvp_i \wedge \psi^*_i
    \]
    and
    \[
 \exec{f(\cvp_1 \wedge \psi^*_1, \ldots
    \cvp_k \wedge \psi^*_k)}^\# \implies \psi_p
    \]
    therefore
    \[
        \exec{e}^{\cP \cup \cP_r} x \implies \psi_p
    \]

    \noindent{\it Induction Hypothesis:} For all expressions $e$ of height less
    than equal to $n$, the following is true:
    \[
        \exec{e}^{\cP \cup \cP_r} x \implies \psi_p
    \]

    \noindent{\it Induction Step:} Let $e = f(e_1, \ldots e_k)$ be an expression
    of height equal to $n+1$.
    The height of expressions $e_1, \ldots e_k$ is less than equal to $n$.

    Note that $\cvp_i \wedge \psi^*_i \implies \exec{e_i}x$ (line-7 and line-9).
    And since $\mathsf{BackPropogate}(e_i,
    x, \cvp_i \wedge \psi^*_i, \cP, \cU) \subseteq \cP_r$, using induction
    hypothesis:
    \[
        \exec{e_i}^{\cP \cup \cP_r} \implies \cvp_i \wedge \psi^*_i
    \]
    and
    \[
 \exec{f(\cvp_1 \wedge \psi^*_1, \ldots
    \cvp_k \wedge \psi^*_k)}^\# \implies \psi_p
    \]
    therefore
    \[
        \exec{e}^{\cP \cup \cP_r} x \implies \psi_p
    \]
\end{proof}

\subsection{Appendix: Non Noisy Performance Comparison}
\label{sec:appendixresults}
%\begin{figure}
\begin{figure}
    {\tiny
  \centering
    \begin{tabular}{|l|c|c|c|c|c|c|c|}
        \hline
    \multirow{2}[0]{*}{Benchmark} & \multirow{2}[0]{*}{No of Examples} &
        \multicolumn{5}{c|}{Rose}              & \multicolumn{1}{c|}{CFTA} \\
          \cline{3-8}
          &       & $L_{0/\infty}$ & $L_{0/1}$   & $L_{DL}$ & $L_{1D}$ & $L_{nS}$ & Threshold 4\\

    \hline
        bikes & 6     & 67    & 67    & 69    & 73    & 67    & 19554 \\
    bikes-long & 24    & 115   & 116   & 125   & 164   & 113   & 58187 \\
    bikes-long-repeat & 58    & 198   & 196   & 230   & 357   & 204   & 127214 \\
    bikes\_small & 6     & 67    & 67    & 73    & 69    & 67    & 21210 \\
    dr-name & 4     & 376   & 408   & 349   & 537   & 378   & - \\
    dr-name-long & 50    & 444   & 501   & 659   & 880   & 446   & - \\
    dr-name-long-repeat & 150   & 788   & 775   & 1123  & 1012  & 845   & - \\
    dr-name\_small & 4     & 381   & 411   & 363   & 531   & 406   & - \\
    firstname & 4     & 121   & 124   & 131   & 135   & 118   & 4258 \\
    firstname-long & 54    & 334   & 332   & 404   & 735   & 333   & 37946 \\
    firstname-long-repeat & 204   & 823   & 811   & 948   & 1080  & 819   & 148101 \\
    firstname\_small & 4     & 120   & 119   & 305   & 170   & 118   & 4220 \\
    initials & 4     & 364   & 340   & 410   & 516   & 199   & 36188 \\
    initials-long & 54    & 626   & 693   & 967   & 848   & 488   & 378070 \\
    initials-long-repeat & 204   & 1252  & 1199  & 1637  & 1627  & 1159  & - \\
    initials\_small & 4     & 347   & 344   & 736   & 422   & 201   & 30920 \\
    lastname & 4     & 117   & 123   & 132   & 1069  & 125   & 175762 \\
    lastname-long & 54    & 341   & 345   & 383   & 426   & 350   & 565654 \\
    lastname-long-repeat & 204   & 825   & 833   & 970   & 1120  & 784   & - \\
    lastname\_small & 4     & 119   & 120   & 139   & 127   & 119   & 178825 \\
    name-combine & 6     & 1336  & 1248  & 1369  & 2231  & 740   & - \\
    name-combine-2 & 4     & 2112  & 2140  & 2119  & 3756  & 1730  & - \\
    name-combine-2-long & 54    & 2514  & 2557  & 1752  & 4748  & 2251  & - \\
    name-combine-2-long-repeat & 204   & 3690  & 3615  & 2567  & 5456  & 2990  & - \\
    name-combine-2\_short & 4     & 2101  & 2101  & 1760  & 2731  & 1748  & - \\
    name-combine-3 & 6     & 303   & 309   & 462   & 1382  & 299   & 547447 \\
    name-combine-3-long & 50    & 475   & 528   & 800   & 665   & 480   & - \\
    name-combine-3-long-repeat & 200   & 982   & 1126  & 1163  & 1389  & 982   & - \\
    name-combine-3\_short & 6     & 298   & 330   & 345   & 375   & 312   & 544044 \\
    name-combine-4 & 5     & 1906  & 1888  & 2253  & 4029  & 1882  & - \\
    name-combine-4-long & 50    & 2286  & 2280  & 3235  & 2668  & 2253  & - \\
    name-combine-4-long-repeat & 200   & 3053  & 3239  & 3436  & 3856  & 2915  & - \\
    name-combine-4\_short & 5     & 1964  & 1868  & 1779  & 3281  & 2018  & - \\
    name-combine-long & 50    & 1633  & 1592  & 2332  & 2867  & 1119  & - \\
    name-combine-long-repeat & 204   & 3590  & 3562  & 4844  & 5079  & 2721  & - \\
    name-combine\_short & 6     & 1365  & 1264  & 1130  & 2542  & 736   & - \\
    phone & 6     & 70    & 67    & 76    & 110   & 68    & 943 \\
    phone-1 & 6     & 59    & 62    & 64    & 72    & 62    & 933 \\
    phone-1-long & 100   & 282   & 280   & 304   & 322   & 279   & 8173 \\
    phone-1-long-repeat & 400   & 727   & 730   & 753   & 811   & 760   & 28547 \\
    phone-10 & 7     & 91284 & 79267 & 141671 & 134173 & 31920 & - \\
    phone-10-long & 100   & 90668 & 72924 & 115134 & 124326 & 35296 & - \\
    phone-10-long-repeat & 400   & 97173 & 88929 & 154668 & 133886 & 38224 & - \\
    phone-10\_short & 7     & 78872 & 76274 & 117722 & 124704 & 28693 & - \\
    phone-1\_short & 6     & 60    & 61    & 68    & 81    & 62    & 942 \\
    phone-2 & 6     & 79    & 81    & 78    & 157   & 70    & 953 \\
    phone-2-long & 100   & 299   & 299   & 543   & 430   & 294   & 6849 \\
    phone-2-long-repeat & 400   & 765   & 794   & 873   & 994   & 773   & 29766 \\
    phone-2\_short & 6     & 79    & 81    & 79    & 89    & 71    & 943 \\
    phone-3 & 7     & 699   & 773   & 694   & 883   & 356   & - \\
    phone-3-long & 100   & 1229  & 1291  & 1639  & 1707  & 774   & - \\
    phone-3-long-repeat & 400   & 2288  & 2406  & 2668  & 2674  & 1954  & - \\
    phone-3\_short & 7     & 808   & 791   & 904   & 1149  & 363   & - \\
    phone-4 & 6     & 2779  & 2511  & 2251  & 4568  & 1089  & - \\
    phone-4-long & 100   & 3529  & 3097  & 4367  & 5388  & 2228  & - \\
    phone-4-long-repeat & 400   & 5514  & 4487  & 5921  & 7886  & 3340  & - \\
    phone-4\_short & 6     & 2578  & 2527  & 2466  & 5311  & 1108  & - \\
\hline    
\end{tabular}%
}
    \caption{Runtime performance of \sys~and CFTA over noise-free dataset}
\label{table:nonoise1}
\end{figure}
\begin{figure}
    {\tiny
  \centering 
\begin{tabular}{|l|c|c|c|c|c|c|c|}
        \hline
    \multirow{2}[0]{*}{Benchmark} & \multirow{2}[0]{*}{No of Examples} &
        \multicolumn{5}{c|}{Rose}              & \multicolumn{1}{c|}{CFTA} \\
          \cline{3-8}
          &       & $L_{0/\infty}$ & $L_{0/1}$   & $L_{DL}$ & $L_{1D}$ & $L_{nS}$ & Threshold 4\\

   \hline
     phone-5 & 7     & 110   & 116   & 100   & 267   & 98    & 122 \\
    phone-5-long & 100   & 407   & 407   & 485   & 466   & 386   & 683 \\
    phone-5-long-repeat & 400   & 1156  & 1153  & 1308  & 1344  & 1123  & 2179 \\
    phone-5\_short & 7     & 109   & 115   & 158   & 126   & 101   & 127 \\
    phone-6 & 7     & 168   & 170   & 148   & 194   & 119   & 3230 \\
    phone-6-long & 100   & 493   & 503   & 484   & 1047  & 447   & 27566 \\
    phone-6-long-repeat & 400   & 1246  & 1295  & 1353  & 1507  & 1234  & 100241 \\
    phone-6\_short & 7     & 169   & 178   & 183   & 231   & 121   & 3327 \\
    phone-7 & 7     & 152   & 163   & 370   & 179   & 123   & 2793 \\
    phone-7-long & 100   & 458   & 485   & 502   & 655   & 438   & 27770 \\
    phone-7-long-repeat & 400   & 1245  & 1253  & 1469  & 1921  & 1230  & 89535 \\
    phone-7\_short & 7     & 162   & 164   & 197   & 199   & 120   & 2762 \\
    phone-8 & 7     & 155   & 162   & 153   & 266   & 114   & 3464 \\
    phone-8-long & 100   & 460   & 453   & 1490  & 538   & 439   & 22961 \\
    phone-8-long-repeat & 400   & 1242  & 1265  & 1483  & 1907  & 1229  & 88932 \\
    phone-8\_short & 7     & 156   & 156   & 146   & 183   & 119   & 3223 \\
    phone-9 & 7     & 30087 & 28603 & 43063 & 67994 & 12823 & - \\
    phone-9-long & 100   & 31448 & 33337 & 48872 & 59942 & 13498 & - \\
    phone-9-long-repeat & 400   & 39881 & 41229 & 56789 & 73652 & 18845 & - \\
    phone-9\_short & 7     & 27348 & 31304 & 43994 & 62737 & 12140 & - \\
    phone-long & 100   & 290   & 285   & 1092  & 329   & 285   & 8592 \\
    phone-long-repeat & 400   & 720   & 752   & 806   & 867   & 766   & 28836 \\
    phone\_short & 6     & 68    & 66    & 68    & 79    & 68    & 963 \\
    reverse-name & 6     & 783   & 732   & 764   & 1186  & 463   & - \\
    reverse-name-long & 50    & 1075  & 1077  & 2120  & 1341  & 764   & - \\
    reverse-name-long-repeat & 200   & 1682  & 1659  & 1912  & 2590  & 1398  & - \\
    reverse-name\_short & 6     & 767   & 774   & 823   & 987   & 440   & - \\
    univ\_1 & 6     & -     & -     & 203455 & -     & 51541 & - \\
    univ\_1-long & 20    & -     & -     & -     & -     & 64459 & - \\
    univ\_1-long-repeat & 30    & 22101 & 22452 & -     & 54745 & 11928 & - \\
    univ\_1\_short & 6     & -     & -     & 202296 & -     & 50989 & - \\
    univ\_3 & 6     & 6500  & 6362  & 5601  & 11083 & 5770  & - \\
    univ\_3-long & 20    & -     & -     & -     & -     & -     & - \\
    univ\_3-long-repeat & 30    & -     & -     & -     & -     & -     & - \\
    univ\_3\_short & 6     & 6697  & 6275  & 4506  & 9127  & 5772  & - \\
   \hline
    \end{tabular}%
    }
    \caption{Runtime performance of \sys~and CFTA over noise-free dataset}
\label{table:nonoise2}
\end{figure}

\end{document}